\documentclass[12pt,reqno]{amsart}

\usepackage{amssymb,amsthm,amsmath,mathrsfs,}
\usepackage{enumerate}
\usepackage{fullpage}

\numberwithin{equation}{section}

\begin{document}

\newtheorem{theorem}{Theorem}[section]
\newtheorem{lemma}[theorem]{Lemma}
\newtheorem{define}[theorem]{Definition}
\newtheorem{remark}[theorem]{Remark}
\newtheorem{corollary}[theorem]{Corollary}
\newtheorem{example}[theorem]{Example}
\newtheorem{assumption}[theorem]{Assumption}
\newtheorem{proposition}[theorem]{Proposition}
\newtheorem{conjecture}[theorem]{Conjecture}

\def\Ref#1{Ref.~\cite{#1}}

\def\Rnum{{\mathbb R}}
\def\pr{\mathop{\rm pr}}
\def\rk{{\rm rank}}
\def\X{\mathrm{X}}
\def\d{\mathrm{d}}
\def\lieder#1{{\mathcal L}_{#1}}

\def\smallbinom#1#2{{\textstyle \binom{#1}{#2}}}

\def\const{\text{const.}}

\def\Esp{\mathcal{E}}
\def\S{S}
\def\Rop{{\mathcal R}}
\def\Dop{{\mathcal D}}
\def\Hop{{\mathcal H}}
\def\Jop{{\mathcal J}}

\def\p{{\rm p}}

\tolerance=50000
\allowdisplaybreaks[3]

\title{A formula for symmetry recursion operators\\ from non-variational symmetries of\\ partial differential equations}

\author{
Stephen C. Anco${}^1$
\lowercase{\scshape{and}}
Bao Wang${}^2$
\\\\{\scshape{
${}^{1,2}$Department of Mathematics and Statistics\\
Brock University\\
St. Catharines, ON L2S3A1, Canada}}
\\{\scshape{
${}^2$School of Mathematics and Statistics\\
Ningbo University\\
Ningbo, 315211, China}}
}

\thanks{${}^1$sanco@brocku.ca, ${}^2$wangbao@nbu.edu.cn}

\begin{abstract}
An explicit formula to find symmetry recursion operators for partial differential equations (PDEs)
is obtained from new results connecting 
variational integrating factors and non-variational symmetries. 
The formula is special case of a general formula that produces a pre-symplectic operator
from a non-gradient adjoint-symmetry. 
These formulas are illustrated by several examples of 
linear PDEs and integrable nonlinear PDEs. 
Additionally, a classification of quasilinear second-order PDEs 
admitting a multiplicative symmetry recursion operator through the first formula is presented. 
\end{abstract}

\maketitle

\section{Introduction}

Symmetry recursion operators have important role \cite{BCA-book,Olv-book} 
in the theory of linear PDEs and nonlinear integrable PDEs. 
For linear PDEs, 
a symmetry recursion operator typically comes from 
the Frechet derivative of a Lie symmetry expressed in characteristic (evolutionary) form. 
In contrast, 
for integrable evolution PDEs, 
a symmetry recursion operator typically comes from 
the ratio of two compatible Hamiltonian operators. 

In the present paper, 
the main result will be to show that a symmetry recursion operator 
can arise in a similar way for nonlinear integrable PDEs
as it does for linear PDEs. 
Specifically, 
a simple explicit formula for a symmetry recursion operator of an Euler--Lagrange PDE
will be presented. 
The formula involves the Frechet derivative of a non-variational symmetry of the PDE, 
and thereby it gives a surprising and interesting new use for non-variational symmetries. 

This result is obtained from a more general result:
an explicit formula for a generalized pre-symplectic operator 
is derived from an adjoint-symmetry of a general PDE 
without needing existence of any variational structure. 
Adjoint-symmetries \cite{AncBlu1997,AncBlu2002b,Anc-review} are the solutions of the adjoint of the determining equation for symmetry vector fields,
and they yield a multiplier for a conservation law of the PDE 
when the product of an adjoint-symmetry and the PDE is annihilated by the Euler--Lagrange operator \cite{BCA-book,Olv-book}. 
The formula to obtain a generalized pre-symplectic operator 
utilizes adjoint-symmetries that are not multipliers. 

In case of Euler--Lagrange PDEs, 
a generalized pre-symplectic operator is the same as a symmetry recursion operator,
and a non-multiplier adjoint-symmetry is the same as a non-variational symmetry. 
For PDEs without any variational structure, 
a generalized pre-symplectic operator \cite{KerKraVer2004} 
is a linear differential operator (in total derivatives) that maps symmetries into adjoint-symmetries, 
analogously to the case of Hamiltonian evolution equations \cite{Olv-book,Dor-book}
where a symplectic operator is a linear differential operator (in total derivatives) that maps 
Hamiltonian symmetries into adjoint-symmetries that are Hamiltonian gradients. 

Our starting point for deriving the formula for a generalized pre-symplectic operator 
comes from a new result that connects variational integrating factors to 
multiplicative pre-symplectic operators. 
A variational integrating factor is a function whose product with a PDE yields an Euler--Lagrange equation. 

All of these new results will be illustrated through examples: 
(i) a linear wave equation and the Korteweg--de Vries equation in potential form
will be used to show how a symmetry recursion operator can be derived 
from a higher symmetry for Euler--Lagrange PDEs; 
(ii) adjoint-symmetries for a spatially-inhomogeneous Airy equation 
and a generalized Korteweg--de Vries equation with $p$-power nonlinearity
will be used to obtain symplectic operators in the case of evolution PDEs; 
(iii)  a 2D Boussinesq equation and a family of peakon equations 
containing the Camassa--Holm equation and its modified version
will be used to show how generalized pre-symplectic operators 
arise from adjoint-symmetries for nonlinear PDEs that neither have an evolution form nor an Euler--Lagrange form. 

The rest of the paper is organized as follows. 

In Section~\ref{sec:prelims}, 
some preliminaries that will be needed concerning 
symmetries, adjoint-symmetries, multipliers, and variational integrating factors 
are summarized. 

In Section~\ref{sec:variationalIFs}, 
the new connection between variational integrating factors 
and multiplicative pre-symplectic operators is derived for general PDEs. 
This connection is further illustrated by giving a classification of 
quasilinear second-order PDEs that admit a variational integrating factor. 

The main results are presented in Section~\ref{sec:mainresults}. 
First, the explicit formula yielding a generalized pre-symplectic operator 
from a non-multiplier adjoint-symmetry is derived. 
A necessary and sufficient condition for this formula to produce 
a conservation law multiplier from a symmetry is also obtained. 
Next, these new results are specialized to evolution PDEs. 
The formula is shown to have a particularly simple form given by 
the non-self-adjoint part of the Frechet derivative of an adjoint-symmetry,
and the condition for producing a conservation law multiplier is shown to be given by 
a Lie derivative condition. 
Euler--Lagrange PDEs are then considered. 
The formula is shown to yield a symmetry recursion operator. 
Here the condition for producing a conservation law multiplier 
coincides with the condition that the formula maps variational symmetries into variational symmetries. 

In Section~\ref{sec:examples}, 
the examples illustrating these results are presented. 

Finally, some concluding remarks are made in Section~\ref{sec:remarks}. 

Standard index conventions, which we will use hereafter, 
are summarized in the Appendix.

\section{Preliminaries}\label{sec:prelims}

Throughout, 
we consider scalar PDEs $G=0$ of any order $N\geq 1$, with 
$n$ independent variables $x=(x^1,\ldots,x^n)$ 
and a single dependent variable $u$. 
In the case of two independent variables, the notation $x^1=t,x^2=x$ will be used. 

The space of formal solutions $u(x)$ of $G=0$ will be denoted $\Esp$. 
Jet space is the coordinate space $J=(x,u,\partial_xu,\ldots)$. 
A differential function is a function on $J$ given by 
$f(x,u^{(k)})=f(x,u,\partial_xu,\ldots,\partial_x^ku)$
for some finite differential order $k\geq 0$. 
For general background on symmetries, adjoint-symmetries, and multipliers 
in this setting, see \Ref{BCA-book,Olv-book,Anc-review}. 

\subsection{Symmetries}

A \emph{symmetry} in characteristic form is a vector field $\hat \X_P = P\partial_u$ 
where $P$ is a differential function satisfying the determining equation
\begin{equation}\label{symm.deteqn}
\pr\hat\X_P (G)|_\Esp = G'(P)|_\Esp =0 . 
\end{equation}
Here $\pr$ denotes prolongation to an appropriate order,
and a prime denotes the Frechet derivative. 
Under a mild regularity condition \cite{Olv-book,Anc-review} on the differential function $G$, 
which holds for all typical PDEs that arise in applied mathematics and physics, 
symmetries satisfy 
\begin{equation}\label{symm.deteqn.offsolns}
G'(P) =R_P(G)
\end{equation}
holding off of $\Esp$,
where $R_P$ is a linear differential operator in total derivatives 
whose coefficients are non-singular differential functions on $\Esp$. 

A \emph{Lie point symmetry} is given by the vector field 
$P_\p= \eta(x,u) - \xi(x,u)^i u_{x^i}$. 
This characteristic form arises from the infinitesimal generator 
$\X_\p = \xi(x,u)^i \partial_{x^i} + \eta(x,u) \partial_u$
of a one-parameter transformation group on $(x,u)$,
which acts on functions $u=f(x)$ via $\hat\X_{\p} f(x) = P_\p|_{u=f(x)}$. 
The generator $\X_\p$ is related to the vector field $\hat\X_{\p}$ by 
\begin{equation}
\pr\X_\p = \pr\hat\X_{\p} +\xi^i D_{x^i} . 
\end{equation}

If $P=P_\p$ is the characteristic function of a Lie point symmetry,
then $\pr\X_\p(G) = F G$ holds for some differential function $F$, 
because the symmetry acts as a generator of a transformation group on $(x,u)$
which prolongs to a point transformation in jet space (up to any finite order).
Note that the differential order of $F$ cannot exceed the differential order of $G$. 
Off of $\Esp$, 
a Lie point symmetry will have $R_{\p} = F -\xi^i D_{x^i}$.

\subsection{Adjoint-symmetries}

An \emph{adjoint-symmetry} of $G=0$ is a differential function $Q$ that satisfies
the adjoint of the determining equation \eqref{symm.deteqn} for symmetries:
\begin{equation}\label{adjsymm.deteqn}
G'{}^*(Q)|_\Esp =0 . 
\end{equation}
Under the previously mentioned regularity condition on $G$, 
adjoint-symmetries satisfy 
\begin{equation}\label{adjsymm.deteqn.offsolns}
G'{}^*(Q) =R_Q(G)
\end{equation}
holding off of $\Esp$, 
where $R_Q$ is a linear differential operator in total derivatives
whose coefficients are non-singular differential functions on $\Esp$. 

In general, adjoint-symmetries have a geometrical meaning 
connected to 1-form fields as explained in \Ref{AncWan}. 
In the case of evolution PDEs, 
they can be viewed geometrically as $1$-forms $Q\d u$ 
which are invariant under the flow defined by the PDE \cite{Anc-review,AncWan}.

\subsection{Euler--Lagrange operator and variational integrating factors}

The Euler--Lagrange operator $E_u=\delta/\delta u$ is the variational derivative with respect to $u$.
It obeys the product rule 
\begin{equation}
E_u(fg) = f'{}^*(g) + g'{}^*(f) . 
\end{equation}

A PDE $G=0$ is an \emph{Euler--Lagrange equation}, 
$G=E_u(L)$ with a Lagrangian given by a differential function $L$, 
iff the Frechet derivative of $G$ is self-adjoint,  $G'=G'{}^*$ \cite{Olv-book,Anc-review}.
This operator equation can be split into the set of Helmholtz conditions 
involving higher Euler--Lagrange operators \cite{Olv-book,Anc-review}. 

Note that a Lagrangian is unique only up to a total divergence, $L\to L + D_{x^i} A^i$,
since $E_u(D_{x^i} A^i) \equiv 0$ holds for any differential functions $A^i$.

\begin{lemma}\label{lem:L.Gis2ndordquasilinear}
Any Euler--Lagrange PDE $G=E_u(L)=0$ that is quasilinear and of second order
possesses a Lagrangian $L$ that is of first order. 
\end{lemma}

\begin{proof}
A Lagrangian for any such PDE,
$G(x,u^{(2)})=G_1^{ij}(x,u^{(1)})u_{ij} + G_0(x,u^{(1)})=0$,
can be constructed by the homotopy formula \cite{Anc-review}
\begin{equation}\label{L.homotopy}
L=\int_0^1 u G(x,\lambda u^{(2)})\,d\lambda
\end{equation}
which yields
\begin{equation*}
L=L_1^{ij}(x,u^{(1)}) u_{ij}+ L_0(x,u^{(1)})
\end{equation*}
where 
\begin{equation}\label{L.conditions}
L_1^{ij}{}_{u_k u_l} + L_1^{kl}{}_{u_i u_j}
 -\tfrac{1}{2}( L_1^{ik}{}_{u_j u_l} + L_1^{jk}{}_{u_i u_l} + L_1^{il}{}_{u_j u_k} + L_1^{jl}{}_{u_i u_k} ) =0 
\end{equation}
due to $G$ being quasilinear and of second order.
We can solve equation \eqref{L.conditions} by tensor algebra methods
as follows.

First we antisymmetrize over the indices $j,k$.
This yields
\begin{equation}\label{L.conditions1}
L_1^{ij}{}_{u_k u_l} - L_1^{ik}{}_{u_j u_l} - L_1^{lj}{}_{u_i u_k} + L_1^{kl}{}_{u_i u_j} =0 ,
\end{equation}
where we have used the symmetry $L_1^{ij}{}_{u_k u_l} = L_1^{ji}{}_{u_k u_l}$.
From equation \eqref{L.conditions1}, we see that
$L_1^{ij}{}_{u_k u_l} - L_1^{ik}{}_{u_j u_l}$ is symmetric in the indices $i,l$,
and thus it must have a gradient form with respect to $u_i$ and $u_l$:
\begin{equation}\label{L.conditions2}
L_1^{ij}{}_{u_k u_l} - L_1^{ik}{}_{u_j u_l}
= \tfrac{1}{2}( L_2^{j}{}_{u_i u_k u_l} - L_2^{k}{}_{u_i u_j u_l} )
\end{equation}
for some differential function $L_2^{j}(x,u^{(1)})$.
The solution of equation \eqref{L.conditions2} can be shown to be given by 
\begin{equation*}
L_1^{ij}{}_{u_k u_l} = \tfrac{1}{2}( L_2^{i}{}_{u_j u_k u_l} + L_2^{j}{}_{u_i u_k u_l} ) 
\end{equation*}
up to a pure gradient term that can be absorbed into the r.h.s.\ terms. 
Next we integrate with respect to the variables $u_k$ and $u_l$, yielding
\begin{equation*}
L_1^{ij} = \tfrac{1}{2}(L_2^{i}{}_{u_j} + L_2^{j}{}_{u_i}) + L_3^{ijk} u_k + L_4^{ij}
\end{equation*}
for some functions $L_3^{ijk}(x,u)$ and $L_4^{ij}(x,u)$
which are symmetric in their indices $i,j$.

Hence, we have
\begin{equation}\label{L.form}
L=( L_2^{i}{}_{u_j}(x,u^{(1)}) + L_3^{ijk}(x,u) u_k + L_4^{ij}(x,u) )u_{ij} + L_0(x,u^{(1)}) .
\end{equation}
We can now use the gauge freedom of adding a total divergence $D_{x^i} A^i(x,u^{(1)})$ to eliminate all of the second-order terms in $L$. 
Put
\begin{equation*}
A^i = -L_2^{i} -(L_3^{ijk} - \tfrac{1}{2}L_3^{jki}) u_j u_k -L_4^{ij}u_j,
\end{equation*}
whence we have 
\begin{equation*}
\begin{aligned}
D_{x^i} A^i = & 
-L_2^{i}{}_{u_j} u_{ij} - L_2^{i}{}_{u} u_{i} - L_2^{i}{}_{x^i}
-L_3^{ijk}u_{ij} u_k - \tfrac{1}{2} L_3^{ijk}{}_{u} u_i u_j u_k
\\&\quad
- (L_3^{ijk}{}_{x^i} - \tfrac{1}{2} L_3^{jki}{}_{x^i}) u_j u_k
-L_4^{ij}u_{ij} - L_4^{ij}{}_{u} u_i u_j - L_4^{ij}{}_{x^i} u_j .
\end{aligned}
\end{equation*}  
This yields the equivalent Lagrangian 
\begin{equation*}
\tilde L =L+ D_{x^i} A^i  =
L_0 - L_2^{i}{}_{u} u_{i} - L_2^{i}{}_{x^i}
- \tfrac{1}{2} L_3^{ijk}{}_{u} u_i u_j u_k
- (L_3^{ijk}{}_{x^i} - \tfrac{1}{2} L_3^{jki}{}_{x^i}) u_j u_k
- L_4^{ij}{}_{u} u_i u_j - L_4^{ij}{}_{x^i} u_j
\end{equation*}
which is of first order. 
\end{proof}

A \emph{variational integrating factor} of $G$ is a differential function $W\not\equiv0$ 
such that $WG$ has the form of an Euler--Lagrange equation: 
$E_u(L)= WG$, for some Lagrangian $L$.
If $W=\const$, it is called trivial. 

\subsection{Multipliers and conservation laws}

A local \emph{conservation law} of $G=0$ is a continuity equation 
\begin{equation}\label{conslaw}
D_{x^i} \Phi^i |_\Esp =0
\end{equation}
where $\Phi^i$ are differential functions called a conserved current. 
In the situation where one independent variable is time $t$, 
the time and space components of the conserved current are called the density and the fluxes, respectively.  

If the conserved current has the form of a curl, 
$\Phi^i|_\Esp = D_{x^j}\Theta^{ij}|_\Esp$ for some differential functions $\Theta^{ij}=-\Theta^{ji}$ 
which are components of a skew tensor, 
then the conservation law \eqref{conslaw} is said to be trivial 
because it contains no useful information about the solutions of $G=0$. 
Two conservation laws are considered to be equivalent if they differ by a trivial conservation law. 

Under the same regularity conditions on $G$ already mentioned, 
a non-trivial conservation law \eqref{conslaw} can be expressed in the form 
\begin{equation}
D_{x^i} \tilde\Phi^i = QG
\end{equation}
holding off of $\Esp$, 
where $\tilde\Phi^i$ is equivalent to $\Phi^i$, 
namely $(\tilde\Phi^i - \Phi^i)|_\Esp$ is a curl. 
Here $Q$ is a differential function called the \emph{multiplier} of the conservation law. 

The determining equation for multipliers is simply
\begin{equation}
E_u(QG) \equiv 0
\end{equation}
which holds off of $\Esp$. 
Using the product rule for the Euler--Lagrange operator, 
this determining equation splits into the adjoint-symmetry equation \eqref{adjsymm.deteqn} for $Q$ on $\Esp$, 
plus a system of Helmholtz-type equations for $Q$ 
involving higher Euler operators \cite{Olv-book,Anc-review}. 

\subsection{Variational symmetries} 

For an Euler--Lagrange PDE $G=E_u(L)=0$, 
a symmetry of the Lagrangian $L$ is a vector field $\hat\X_P = P\partial_u$ such that 
$\pr\hat\X_P(L) = L'(P) = D_{x^i} A^i$ 
for some differential functions $A^i$. 
These symmetries are called \emph{variational}. 
More succinctly, variational symmetries are the solutions of the determining equation 
$E_u(L'(P))= E_u(GP)\equiv 0$. 

Thus $P$ will be the characteristic function of a variational symmetry 
iff $P$ is a conservation law multiplier. 
This is the statement of Noether's theorem. 

\subsection{Generalized pre-symplectic operators and pre-Hamiltonian operators}

In the terminology of \Ref{KerKraVer2004}, 
a generalized pre-symplectic operator is a linear differential operator (in total derivatives)
that maps symmetries into adjoint-symmetries,
and likewise a generalized pre-Hamiltonian operator is a linear differential operator (in total derivatives) 
that maps adjoint-symmetries into symmetries. 
The composition of these operators yields a recursion operator on (adjoint-) symmetries.

\section{Results on Variational Integrating Factors}\label{sec:variationalIFs}

As a starting point, for a general PDE $G=0$, 
a simple relationship between variational integrating factors, symmetries, and adjoint-symmetries 
will be presented. 

\begin{proposition}\label{prop:varIF.symm-to-adjsymm}
If $G=0$ admits a variational integrating factor $W$, 
then for any symmetry $\hat\X_P=P\partial_u$ (in characteristic form) of $G=0$, 
there is an adjoint-symmetry $Q= W P$. 
\end{proposition}

\begin{proof}
$WG$ is Euler--Lagrange iff 
\begin{equation}\label{WG.selfadjoint}
WG'+GW'=G'{}^* W +W'{}^* G .
\end{equation}
First, rearrange this operator equation into the form 
$G'{}^* W = WG'+GW' -W'{}^* G$. 
Next, apply it to any symmetry characteristic function $P$:
$G'{}^*(WP)= WG'(P)+GW'(P) -W'{}^*(GP)$. 
Hence, on $\Esp$, we get
\begin{equation}
G'{}^*(WP)|_\Esp =WG'(P)|_\Esp = 0 
\end{equation}
when $P\partial_u$ is a symmetry. 
\end{proof}

The mapping $P\to Q=WP$ can be iterated when $G=0$ is Euler--Lagrange, 
since adjoint-symmetries coincide with symmetries. 

\begin{proposition}\label{prop:varIF.symm-to-symm}
If $G=0$ is Euler--Lagrange and admits a variational integrating factor $W$, 
then starting from any symmetry $\hat\X_P=P\partial_u$ (in characteristic form) of $G=0$, 
there is a sequence of symmetries $\hat\X^{(l)} = W^l P\partial_u$, $l=0,1,2,\ldots$. 
Thus, $W$ is a multiplicative recursion operator for symmetries. 
\end{proposition}

\begin{proof}
Proposition~\ref{prop:varIF.symm-to-adjsymm} shows that $WP$ will be a symmetry of $G=0$. 
Hence $W$ is a multiplicative recursion operator for symmetries. 
\end{proof}

$W=1$ is a trivial example. 
More generally, $W=\const$ is trivial as a recursion operator,
and so we will be interested to find examples in which $W$ is not a constant. 

To begin, 
we derive a necessary condition for existence of a non-trivial variational integrating factor 
when $G=0$ is Euler--Lagrange. 

The determining equation for $W$ is the operator equation \eqref{WG.selfadjoint},
which reduces to 
\begin{equation}\label{Wdeteqn}
(G' W-WG')|_\Esp=0 
\end{equation}
when $G$ is Euler--Lagrange. 
In this determining equation \eqref{Wdeteqn}, 
the terms without derivatives on $W$ will cancel. 
So we have 
\begin{equation}\label{Wdeteqn.lhs}
G'W -WG' = \sum_{\substack{J\\ |K|<|J|\leq N}} \smallbinom{J}{K} (G_{u^J})(D_{J/K} W) D_{K}
\end{equation}
where $K$ is a free multi-index with $0\leq|K|\leq N-1$. 
Since this expression \eqref{Wdeteqn.lhs} must vanish on $\Esp$, 
we find that the highest-order terms $D_{K}$ yield
\begin{equation}\label{Weqn.highestord}
\big( \sum_{i} \smallbinom{K,i}{i} (G_{u_{i,K}}) (D_{i} W) \big)|_\Esp =0,
\quad
|K|=N-1 .
\end{equation}
This is a set of linear equations on $\{D_i W\}$, $i=1,2,\ldots n$, 
which is indexed by $K = \{k_1,\ldots,k_{N-1}\}$. 
Viewing $D_i W$ as the components of a column vector, 
we see that the coefficient matrix 
\begin{equation}\label{DGmatrix}
\begin{pmatrix}
(\binom{K,i}{i} G_{u_{i,K}})|_\Esp
\end{pmatrix}
\end{equation}
defined by this set of linear equations \eqref{Weqn.highestord}
has rows indexed by $\{k_1,\ldots,k_{N-1}\}$
and columns indexed by $i$. 
Its size is $\binom{n+N-2}{n-1}\times n$. 

Existence of a non-constant solution for $W$
requires that the rank of the coefficient matrix \eqref{DGmatrix} 
is not maximal.
This leads to the following necessary condition on $G$. 

\begin{proposition}\label{prop:rankDGmatrix}
An Euler--Lagrange PDE $G=E_u(L)=0$ 
admits a non-constant variational integrating factor
only if 
\begin{equation}\label{rankDGmatrix}
\rk
\begin{pmatrix}
(\binom{i_1i_2\ldots i_N}{i_1} G_{u_{i_1i_2\ldots i_N}})|_\Esp
\end{pmatrix}
\leq n-1 .
\end{equation}
\end{proposition}

The rank condition \eqref{rankDGmatrix} constitutes an equation on the form of $G$. 
It can be expressed in a simple, more explicit way 
in the case when $G$ is of second order. 

For $N=2$, 
the matrix \eqref{DGmatrix} has size $n\times n$,
and hence the rank condition is equivalent to the vanishing of the determinant 
\begin{equation}\label{detDGmatrix}
%\begin{vmatrix}
%2G_{u_{1\,1}\mathstrut} & \cdots &  G_{u_{1\,n}\mathstrut} \\
%\vdots & \ddots & \vdots \\
%G_{u_{n\,1}\mathstrut} & \cdots & 2 G_{u_{n\,n}\mathstrut}
%\end{vmatrix}_\Esp
\det\big( (\delta_{ij} + 1) G_{u_{ij}} \big)_\Esp
=0
\end{equation}
where the coefficients in the matrix are equal to 
$2$ on the diagonal and $1$ on the off-diagonal. 
Moreover, this determinant equation has even a nicer form when $G$ is quasilinear.
By Lemma~\ref{lem:L.Gis2ndordquasilinear}, 
the Lagrangian for $G$ can be assumed to have the form $L(x,u^{(1)})$,
so that 
$G = E_u(L) = L_{u} - L_{x^i u_i} -  u_i L_{uu_i} - u_{ij}L_{u_j u_i}$. 
Then we have $G_{u_{ij}} = -L_{u_{i} u_{j}}$,
which can be used to express the determinant equation \eqref{detDGmatrix} 
entirely in terms of $L$:
\begin{equation}\label{detDLmatrix}
%\begin{vmatrix}
%L_{u_{1\mathstrut} u_{1\mathstrut}} & \cdots & L_{u_{1\mathstrut} u_{n\mathstrut}} \\
%\vdots & \ddots & \vdots \\
%L_{u_{\mathstrut n} u_{1\mathstrut}} & \cdots & L_{u_{n\mathstrut} u_{n\mathstrut}}
%\end{vmatrix}_\Esp
\det\big( L_{u^iu^j} \big)
=0 . 
\end{equation}
Thus, we have established the following result. 

\begin{proposition}\label{prop:detDGmatrix.Gis2ndordquasilinear.MAeqn.L}
If a quasilinear second-order Euler--Lagrange PDE $G=E_u(L)$
admits a non-constant variational integrating factor, 
then the first-order Lagrangian satisfies the Monge--Ampere equation \eqref{detDLmatrix}. 
\end{proposition}

For a general Euler--Lagrange PDE,
we will now explain how, surprisingly, 
existence of a non-variational Lie point symmetry of $G=E_u(L)=0$ 
leads to an explicit formula for a variational integrating factor. 

\begin{proposition}\label{prop:EL.recursion}
Suppose an Euler--Lagrange equation $G=0$ possesses 
a non-variational Lie point symmetry 
$\X_\p=\xi(x,u)^i\partial_{x^i} +\eta(x,u)\partial_u$. 
Then $G=0$ possesses a variational integrating factor 
\begin{equation}\label{W}
W=f_\p +\eta_u+\xi^i_{x^i}
\end{equation}
with $f_\p$ defined by 
\begin{align}\label{X.on.G}
\pr \X_\p(G)=f_\p G . 
\end{align}
If this function \eqref{W} is non-constant, 
then it defines a non-trivial multiplicative recursion operator for the symmetries of $G=0$.
\end{proposition}

\begin{proof}
The characteristic form of $\X_\p$ is 
\begin{equation}\label{P}
\hat\X_\p=P_\p\partial_u, 
\quad
P_\p=\eta -\xi^i u_{x^i} . 
\end{equation}
Hence 
\begin{align}\label{G'Pterm}
\pr\hat\X_\p(G)=(f_\p -\xi^i D_{x^i})G = G'(P_\p) . 
\end{align}
Now, consider the equation
\begin{equation}\label{ELterms}
E_u(\pr\hat\X_\p(L))
= E_u(P_\p G)=P_\p'{}^*(G)+G'(P_\p)
\end{equation}
where
\begin{equation}\label{P'*term}
\begin{aligned}
P_\p'{}^*(G) 
= & (\eta -\xi^i u_{x^i})'{}^*(G)\\
= & \eta_u G  -u_{x^i}\xi^i_u G  + D_{x^i}(\xi^i G)\\
= & (\eta_u + \xi^i_{x^i}) G +\xi^iD_{x^i}G . 
\end{aligned}
\end{equation}
Substitution of \eqref{P'*term} and \eqref{G'Pterm} into equation \eqref{ELterms}
yields 
\begin{equation}
P_\p'{}^*(G)+G'(P_\p)
=WG = E_u(\pr\hat\X_\p(L)) . 
\end{equation}
This shows, firstly, that $WG$ is Euler--Lagrange, and secondly, that $W$ is non-zero 
when $\pr\hat\X_\p(L)$ is not a total divergence. 
Hence, if the symmetry $\hat\X_\p$ is non-variational, 
then $W$ is a variational integrating factor. 
\end{proof}

We will now apply this result to some common geometrical examples of Lie point symmetries.

\emph{Translation symmetry}: 
$\X = a^i\partial_{x^i}$,
with $a^i=\const$ being the components of a vector in $\Rnum^n$. 
If $G$ is invariant off $\Esp$, namely $\X G =0$, then $f_\p=0$ yields 
$W=0$ which is trivial. 

\emph{Scaling symmetry}: 
$\X = \nu u\partial_u + \mu_{(i)} x^i\partial_{x^i}$,
with $\nu$ and $\mu_{(i)}$ being (constant) scaling weights. 
If $G$ has scaling weight $\omega$, namely $\X G = \omega G$ off $\Esp$,
then $W= \omega+\nu +\sum_{i}\mu_{(i)}$ is a constant. 

\emph{Projective symmetry}:
$\X = p \Omega(x) u\partial_u + a^i |x|^2_g \partial_{x^i}$, 
$\Omega(x)= a^i (|x|^2_g)_{x^i}= 2 a^i g_{ij}x^j$, 
$|x|_g^2 = g_{jk}x^j x^k$, 
with $a^i=\const$ being the components of a vector in $\Rnum^n$,
and with $g_{ij}=\const$ being the components of metric on $\Rnum^n$. 
If $G$ has projective weight $q$, namely $\X G = q\Omega G$ off $\Esp$,
then $W= (p+q+1)\Omega$ is a linear function of $x^i$. 
This is a non-trivial variational integrating factor whenever $p+q+1\neq 0$. 
%$\X = 2p a^i g_{ij}x^j u\partial_u + a^i g_{jk}x^j x^k \partial_{x^i}$, 

As we will show next, 
besides projective symmetries, there are many non-geometric examples.

\subsection{Classification}

We now state a general classification of non-trivial variational integrating factors
for quasilinear second-order Euler--Lagrange equations in two independent variables 
(namely, $N=n=2$). 

\begin{theorem}\label{thm:classification}
Suppose 
\begin{equation}\label{quasilinear.2ndord.pde}
G(u^{(2)})=A_1(u^{(1)}) u_{tt} + A_2(u^{(1)}) u_{tx} + A_3(u^{(1)}) u_{xx} +A_0(u^{(1)}) =0
\end{equation}
is an Euler--Lagrange PDE, which is quasilinear and translation invariant. 
If it admits a variational integrating factor $W(t,x,u^{(2)})$ that is not a constant, 
then it is equivalent (modulo a point transformation) to one of the following PDEs:
\begin{align} 
& \begin{aligned}
(a)\quad& 
A_1 = \frac{f(u,v) v^2}{u_t}, 
\quad
A_2 = -2\frac{f(u,v) v}{u_t}, 
\quad
A_3 = \frac{f(u,v)}{u_t}, 
\quad
A_0 = 1 \text{ or } 0, 
\\
& 
v=\frac{u_x}{u_t}; 
\quad
f(u,v) \text{ arbitrary }
\end{aligned}
\\
& \begin{aligned}
(b)\quad & 
A_1 = \frac{f_{v}(u,v) g(u) ^2}{v},  
\quad
A_2 = 2\frac{f_{v}(u,v) g(u) }{v}, 
\quad
A_3 = \frac{f_{v}(u,v)}{v}, 
\\&
A_0 = f_{u}(u,v) + f_{v}(u,v) g_{u}(u)  u_t, 
\\
& 
v=g(u)  u_t + u_x; 
\quad
g(u), f(u,v) \text{ arbitrary }
\end{aligned}
\\
& \begin{aligned}
(c)\quad & 
A_1 = \frac{f_{v}(u,v) v^2}{(g_v(u,v)-u_t)(u_x+v u_t)},  
\quad
A_2 = 2\frac{f_{v}(u,v) v}{(g_v(u,v)-u_t)(u_x+v u_t)},  
\\&
A_3 = \frac{f_{v}(u,v)}{(g_v(u,v)-u_t)(u_x+v u_t)},  
\quad
A_0 = f_{u}(u,v) + \frac{f_{v}(u,v) g_{u}(u,v)}{g_v(u,v)-u_t},
\\
& 
v=h(u,u_t,u_x), 
\quad
v u_t + u_x = g(u,v); 
\quad
f(u,v), g(u,v) \text{ arbitrary }
\end{aligned}
\end{align}
The corresponding variational integrating factors are given by 
\begin{align}
& \begin{aligned}
(a)\quad & 
W=F(u,x+h_1(u,v),t- h_2(u,v))
\\&
h_1(u,v) = \int f(u,v)\,dv + k_1(u),
\quad
h_2(u,v) = \int vf(u,v)\,dv + k_2(u),
\end{aligned}
\\
& \begin{aligned}
(b,c)\quad & 
W=F(f(u,v),x-h_1(u,f(u,v)),t- h_2(u,f(u,v))),
\\& 
h_1(u,w)= \int \frac{1}{f^{-1}(u,w)} \,d u + k_1(w),
\quad
h_2(u,w) = \int g(u) h_{1u}(u,w)\,d u + k_2(w),
\\&
f(u,f^{-1}(u,w))=w; 
\quad
k_1(w), k_2(w) \text{ arbitrary }
\end{aligned}
\end{align}
where $F$ is an arbitrary function of its arguments. 
Moreover, 
in each case, 
$W=F$ is a multiplicative recursion operator for symmetries of the corresponding PDE. 
\end{theorem}

The proof will be given in the next subsection. 

\begin{remark}\hfil
\begin{itemize}
\item[1.]
A Lagrangian for the PDEs (a)--(c) can be obtained from either the homotopy formula \eqref{L.homotopy} or integration of the equations
$A_1 = - L_{u_t u_t}$, $A_2 = -2 L_{u_t u_x}$, $A_3 = - L_{u_x u_x}$, 
$A_0=L_u - L_{u u_t}u_t - L_{u u_x}u_x$. 
\item[2.]
Each PDE (a)--(c) possesses a sequence of contact symmetries 
$\hat\X^{(l)} = W^l P\partial_u$, $l=0,1,2,\ldots$, 
starting from the translation symmetry given by $P=a u_t + b u_x$, 
where $a,b$ are arbitrary constants. 
\item[3.]
All of the PDEs (a)--(c) are of parabolic type, 
since their coefficients satisfy the algebraic relation $A_2{}^2 - 4 A_1 A_3 =0$. 
\item[4.]
None of these PDEs can be mapped into the linear parabolic PDE $u_{tt}\pm 2u_{tx} +u_{xx}=0$ 
by a contact transformation. 
\end{itemize}
\end{remark}

To explain the fourth remark, 
consider a contact transformation \cite{BCA-book} on the space $(t,x,u,u_t,u_x)$. 
We apply the transformation to the Lagrangian density 
$L\,dt\,dx=\tfrac{1}{2}(u_t\pm u_x)^2dt\,dx$
for $u_{tt}\pm 2u_{tx} +u_{xx}=0$, 
and then we impose the condition that the form of the resulting Euler--Lagrange PDE is quasilinear and translation invariant. 
This shows that the contact transformation has the form 
$t^*=t + a(u)$, $x^*=x + b(u)$, $u^*=c(u)$, $(u_t)^*=c'(u)u_t/(1+a'(u)u_t +b'(u)u_x)$, 
$(u_x)^*=c'(u)u_t/(1+a'(u)u_t +b'(u)u_x)$,
up to a point transformation, 
where $a$, $b$, $c$ are arbitrary functions of $u$. 
The PDEs ($a$)--($c$) in the classification, however, 
involve the arbitrary function $f(u,v)$ of two variables. 
Therefore, in general, these PDEs cannot be mapped into $u_{tt}\pm 2u_{tx} +u_{xx}=0$
by a contact transformation. 

The PDE in class ($a$) is equivalent to (up to a variational multiplier)
\begin{equation}
u_t v_x - u_x v_t + A_0 u_t/f(u,v) =0 .
\end{equation}
An example of such a PDE is closely related to the Born--Infeld equation \cite{BorInf} in 1+1 dimensions, 
$(a^2+u_x{}^2) u_{tt} -2 u_tu_x u_{tx} -(a^2-u_t{}^2) u_{xx}=0$,
which arises from the Lagrangian $L=\sqrt{u_t{}^2-u_x{}^2 - a^2}$. 
Its dispersionless limit is given by 
\begin{equation}
u_x{}^2 u_{tt} -2 u_tu_x u_{tx} + u_t{}^2 u_{xx}=0, 
\end{equation}
which is equivalent to the Euler--Lagrange equation obtained from $L$ with $a=0$. 
This dispersionless Euler--Lagrange PDE belongs to class ($a$):
$A_1 = u_x{}^2/L^3|_{a=0}$, 
$A_2 = -2u_t u_x/L^3|_{a=0}$, 
$A_3 = u_t{}^2/L^3|_{a=0}$, 
$A_0=0$, 
$f(u,v)=1/\sqrt{1-v^2}^3$,
where $L|_{a=0}=\sqrt{u_t{}^2-u_x{}^2}$ is the Lagrangian.

\subsection{Proof of classification}

A necessary condition for existence of a variational integrating factor of 
a general second-order PDE is given by 
the determinant condition \eqref{detDGmatrix}. 
When the PDE has a quasilinear form \eqref{quasilinear.2ndord.pde},
this determinant becomes 
\begin{equation}\label{det.cond}
A_2{}^2 - 4 A_1 A_3 =0 . 
\end{equation}

The classification now proceeds by setting up the necessary and sufficient conditions
\begin{equation}\label{sys1}
G'=G'{}^*,
\quad
(GW)'=(GW)'{}^*,
\quad
W\neq\const . 
\end{equation}
For the PDE \eqref{quasilinear.2ndord.pde} to be of second order, 
we must have $A_1A_3\neq 0$, since otherwise $A_1=A_3=0$ implies $A_2=0$ 
from the condition \eqref{det.cond}. 
Then without loss of generality, we may assume 
\begin{equation}\label{sys2}
A_1 \not\equiv 0 
\end{equation}
(after permutation of the variables $(t,x)$ if necessary). 
We also will assume 
\begin{equation}\label{cond1}
(A_0{}_{u_x})^2+(A_1{}_{u_x})^2+A_2{}^2+A_3{}^2\not\equiv 0
\end{equation}
which ensures that the PDE \eqref{quasilinear.2ndord.pde} does not reduce to an ODE
due to having no dependence on $x$-derivatives of $u$. 

The conditions \eqref{sys1} can be expressed in terms of higher Euler operators 
$E_u^{(i)}$, $E_u^{(ij)}$ \cite{Olv-book,Anc-review}:
\begin{subequations}\label{GWsys}
\begin{align}
& G_u = E_u(G),
\quad
G_{u_t} = -E_u^{(t)}(G),
\quad
G_{u_x} = -E_u^{(x)}(G),
\\
& (GW)_u = E_u(GW),
\quad
(GW)_{u_t} = -E_u^{(t)}(GW),
\quad
(GW)_{u_x} = -E_u^{(x)}(GW),
\end{align}
\end{subequations}
and
\begin{subequations}\label{empty}
\begin{align}
& 
G_{u_{tt}} = E_u^{(t,t)}(G),
\quad
G_{u_{tx}} = E_u^{(t,x)}(G),
\quad
G_{u_{xx}} = E_u^{(x,x)}(G),
\\
& 
(GW)_{u_{tt}} = E_u^{(t,t)}(GW),
\quad
(GW)_{u_{tx}} = E_u^{(t,x)}(GW),
\quad
(GW)_{u_{xx}} = E_u^{(x,x)}(GW).
\end{align}
\end{subequations}
Equations \eqref{empty} reduce to identities. 
The remaining equations \eqref{GWsys} 
can be split with respect to the jet variables $u_{tt}$, $u_{tx}$, $u_{xx}$. 
This splitting, together with equation \eqref{sys2} and condition \eqref{cond1},
yields a nonlinear overdetermined system on the functions $A_1$, $A_2$, $A_3$, $A_0$, and $W$,
subject to conditions \eqref{sys2} and \eqref{cond1}. 

We solve this system by using two main steps. 
First, we use Maple 'rifsimp' to do an integrability analysis.  
This yields two main cases that are separated by whether 
$u_tA_2+2u_xA_3$ is zero or non-zero. 
Next, in each case we integrate the resulting system. 

\emph{Case $(1)$}:\quad$u_tA_2+2u_xA_3=0$

In this case, the nonlinear overdetermined system consists of 
\begin{equation}\label{sys.case1}
\begin{aligned}
& 
A_0{}_{u_t}=0,
\quad 
A_0{}_{u_x}=0,
\quad
u_tA_2+2u_xA_3=0,
\quad
u_tA_3{}_{u_t}+u_xA_3{}_{u_x}+A_3=0,
\\
& 
u_tW_{u_t}+u_xW_{u_x}=0,
\quad 
(u_xW_t-u_tW_x)A_3 +u_tA_0W_{u_x}=0,
\quad
A_1 = A_2{}^2/(4A_3) . 
\end{aligned}
\end{equation}
The solution of the linear PDEs for $A_0$, $A_2$, $A_3$ is given by 
\begin{equation}\label{case1.A0A3A2}
A_0 = g(u),
\quad
A_2 = -2\frac{f(u,v) v}{u_t}, 
\quad
A_3 = \frac{f(u,v)}{u_t},
\end{equation}
where $v=u_x/u_t$, 
and where $f(u,v)$, $g(u)$ are arbitrary functions. 
Hence
\begin{equation}\label{case1.A1}
A_1 = \frac{f(u,v)v^2}{u_t} .
\end{equation}
From the resulting form of the quasilinear PDE \eqref{quasilinear.2ndord.pde},
we see that if $g(u)$ is not zero then it can be absorbed in $f(u,v)$, 
and hence we can take 
\begin{equation}\label{case1.g}
g(u)=1 \text{ or } 0. 
\end{equation}
Substituting expressions \eqref{case1.A0A3A2} and \eqref{case1.g}
into the remaining two equations in \eqref{sys.case1}, 
we obtain the solution for $W$.
This yields case $(a)$ in Theorem~\ref{thm:classification}. 

\emph{Case $(2)$}:\quad$u_tA_2+2u_xA_3 \neq0$

In this case, the nonlinear overdetermined system consists of
\begin{equation}\label{sys.case2.Aeqns}
\begin{aligned}
& A_1 = (u_x/u_t)^2 A_3 , 
\quad
2A_3{}_{u_t}-A_2{}_{u_x}=0,
\quad
2A_3^2A_2{}_{u_t}+A_2^2A_3{}_{u_x}-2A_2A_3A_2{}_{u_x}=0,
\\
& A_2{}_u(u_tA_2+2u_xA_3)^2-2u_tA_2^2A_0{}_{u_x}-8u_xA_3^2A_0{}_{u_t}=0,
\\
& A_3{}_u(u_tA_2+2u_xA_3)^2-4A_3((u_tA_2+u_xA_3)A_0{}_{u_x}-u_tA_3A_0{}_{u_t})=0,
\\
& 4A_3^2A_0{}_{u_tu_t}-4A_2A_3^2A_0{}_{u_tu_x}+A_2^2A_3A_0{}_{u_xu_x}
-2(A_2A_0{}_{u_x}-2A_3A_0{}_{u_t})(A_2A_3{}_{u_x}-A_3A_2{}_{u_x})=0,
\end{aligned}
\end{equation}
and 
\begin{equation}\label{sys.case2.Weqns}
2A_3W_{u_t}-A_2W_{u_x}=0,\quad
A_2W_t+u_tA_2W_u+2u_xA_3W_u+2A_3W_x-2A_0W_{u_x}=0. 
\end{equation}
To begin, 
we can solve for $A_2$ and $A_3$ 
in the first two equations in \eqref{sys.case2.Aeqns}
by introducing 
\begin{equation}\label{sys2.case2.A4rel}
A_4=\frac{A_2}{2A_3} ,
\end{equation}
whereby these two equations become
\begin{align}
& A_3{}_{u_t}-A_4A_3{}_{u_x}=0,
\label{sys.case2.A3eqn}
\\
& A_4{}_{u_t}-A_4A_4{}_{u_x}=0.
\label{sys.case2.A4eqn}
\end{align}
Equation \eqref{sys.case2.A4eqn} is a first-order nonlinear PDE for $A_4$. 
Its general solution is implicitly determined by
\begin{align*}
\Phi(u,A_4,u_x+u_tA_4)=0
\end{align*}
where $\Phi$ is an arbitrary function. 
Next, the solution of equation \eqref{sys.case2.A3eqn} depends on whether 
$\Phi$ has essential dependence on its third argument or not. 

This leads to a case splitting. 
When $\Phi$ has no essential dependence on $u_x+u_tA_4$,
then 
\begin{equation}\label{case2b.A4}
A_4=g(u) , 
\end{equation}
and otherwise when $\Phi$ does have essential dependence on $u_x+u_tA_4$,
then 
\begin{equation}\label{case2c.A4}
u_tA_4+u_x=g(u,A_4) , 
\end{equation} 
where, in both cases, $g$ is an arbitrary function.

\emph{Case 2$(i)$}: 
In this subcase we can easily solve for $A_3$ from the linear first-order PDE \eqref{sys.case2.A3eqn}, 
which gives 
\begin{equation}\label{case2b.A3}
A_3=F(u,g(u)u_t+u_x)
\end{equation}
where $F$ is an arbitrary function.
Hence
\begin{equation}\label{case2b.A1}
A_1=F(u,g(u)u_t+u_x)(u_x/u_t)^2 . 
\end{equation}
Then, from relation \eqref{sys2.case2.A4rel}, we obtain
\begin{equation}\label{case2b.A2}
A_2=2 g(u)F . 
\end{equation}
To continue with this subcase, 
we now solve the remaining three equations in \eqref{sys.case2.Aeqns} 
for $A_0$ through a change of variables
\begin{align*}
v=u,\quad
v_1=u_t,\quad
v_2=g(u)u_t+u_x.
\end{align*}
The result is 
\begin{equation}\label{case2b.A0}
A_0=v_1 g' f_{v_2} +f_v
\end{equation}
where $f(v,v_2)$ is defined by 
\begin{equation*}
f_{v_2}=v_2F(v,v_2) .
\end{equation*}

Finally, 
substituting expressions \eqref{case2b.A0}, \eqref{case2b.A2}, \eqref{case2b.A1}, \eqref{case2b.A3}, \eqref{case2b.A4} 
into equations \eqref{sys.case2.Weqns}, 
we get 
\begin{align*}
W_{v_1}=0,
\quad 
v_2f_{v}W_{v_2} -f_{v_2}(W_x+g(v)W_t+v_2W_v) =0.
\end{align*}
We can solve this pair of linear first-order PDEs for $W$ 
after another change of variables
\begin{align*}
w=v,\quad
w_1=v_1,\quad
w_2=f.
\end{align*}
The solution yields case $(b)$ in Theorem~\ref{thm:classification}. 

\emph{Case 2$(ii)$}: 
In this subcase the linear first-order PDE \eqref{sys.case2.A3eqn}
for $A_3$ yields 
\begin{equation}\label{case2c.A3}
A_3=\frac{F(u,A_4)}{g_{A_4}-u_t}
\end{equation}
where $F$ is an arbitrary function.
Hence
\begin{equation}\label{case2c.A1}
A_1=\frac{F(u,A_4)u_x^2}{(g_{A_4}-u_t)u_t^2} . 
\end{equation}
Then, from relation \eqref{sys2.case2.A4rel}, we obtain
\begin{equation}\label{case2c.A2}
A_2=2A_4 \frac{F(u,A_4)}{g^{\mathstrut}_ {A_4}-u_t} .
\end{equation}
Continuing similarly to the previous subcase, 
we use a change of variables
\begin{align*}
v=u,\quad
v_1=A_4,\quad
v_2=u_x,
\end{align*}
which allows solving the remaining three equation in \eqref{sys.case2.Aeqns} to obtain
\begin{equation}\label{case2c.A0}
A_0=f_v-\frac{g_vf_{v1}v_1}{v_2+v_1g_{v_1}-g}
\end{equation}
where $f(v,v_1)$ is defined by
\begin{equation*}
f_{v_1}=g(v,v_1)F(v,v_1) .
\end{equation*}

Now, substituting expressions \eqref{case2c.A0}, \eqref{case2c.A2}, \eqref{case2c.A1}, 
\eqref{case2c.A3}, \eqref{case2c.A4} 
into equations \eqref{sys.case2.Weqns}, 
we have 
\begin{align*}
W_{v_2}=0,
\quad 
gf_{v_1}W_v+f_{v_1}W_x+v_1f_{v_1}W_t-gf_{v_1}W_{v_1}=0.
\end{align*}
We can solve this pair of linear first-order PDEs for $W$ 
through a further change of variables
\begin{align*}
w=v,\quad
w_1=f(v,v_1),\quad
w_2=v_2.
\end{align*}
The solution yields case $(c)$ in Theorem~\ref{thm:classification}. 

This completes the proof.

\section{Generalizations and Main Results}\label{sec:mainresults}

The preceding results can be generalized in an interesting way 
to yield non-multiplicative operators. 

\begin{theorem}\label{thm:op.symm-to-adjsymm}
Suppose $G=0$ possesses an adjoint-symmetry $Q$ that is not a multiplier, namely
\begin{align*}
G'{}^*(Q)=R_Q(G),
\quad
E_u(QG)\not\equiv 0
\end{align*}
holds off of $\Esp$,
where $R_Q$ is non-zero linear differential operator in total derivatives whose coefficients are non-singular differential functions on $\Esp$. 
Then 
\begin{align}\label{S.symm-to-adjsymm}
\S := R_Q^* + Q'
\end{align}
is a linear differential operator in total derivatives that maps symmetries to adjoint-symmetries. 
Moreover, when an inverse $\S^{-1}$ exists, 
it maps adjoint-symmetries to symmetries. 
\end{theorem}

\begin{proof}
By the product rule for $E_u$, we have 
$E_u(QG) = G'{}^*(Q) + Q'{}^*(G) = S^*(G)$,
which is assumed to be non-trivial. 
Since $\S^*(G)$ is Euler--Lagrange, it is self-adjoint \cite{Olv-book,Anc-review}
\begin{equation}\label{S*G}
(\S^*(G))' = (\S^*(G))'{}^* . 
\end{equation}
Consider this operator equation applied to an arbitrary differential function $A$. 
By writing $\S^*=\sum_J W^J D_J$, we see that the l.h.s.\ is given by 
\begin{equation}\label{lhs.S*G}
(\S^*(G))'A=  \sum_J W^J D_J(G'(A))+ (W^J)'(A) D_JG . 
\end{equation}
The r.h.s.\ is the adjoint:
\begin{equation}\label{rhs.S*G}
(\S^*(G))'{}^*A = \sum_J (-1)^{|J|} G'{}^*(D_J(A W^J)) +(W^J)'{}^*(AD_JG) . 
\end{equation}
Hence, on $\Esp$, 
we have
\begin{equation}\label{S*Geqn}
\begin{aligned}
& \big( (\S^*(G))'A \big)\big|_\Esp =  \sum_J \big(W^J D_J(G'(A))\big)\big|_\Esp = \S^*(G'(A))|_\Esp 
\\
& = \big( (\S^*(G))'{}^*A \big)\big|_\Esp = \sum_J (-1)^{|J|} \big(G'{}^*(D_J(A W^J))\big)\big|_\Esp = G'{}^*(\S(A))|_\Esp
\end{aligned}
\end{equation}
using $\S=\sum_J (W^J D_J){}^* = \sum_J (-1)^{|J|} D_J W^J$. 
Rearranging this equation \eqref{S*Geqn}, 
and putting  $A=P$, 
we get 
\begin{align*}
G'{}^*(\S(P))|_\Esp = \S^*(G'(P))|_\Esp=0
\end{align*}
when $P$ is the characteristic function of a symmetry of $G=0$. 
\end{proof}

The formula \eqref{S.symm-to-adjsymm} uses adjoint-symmetries $Q$ 
which are lifted off of the solution space $\Esp$ of $G=0$. 
Recall, an adjoint-symmetry is trivial if it vanishes on $\Esp$,
and two adjoint-symmetries that differ by a trivial adjoint-symmetry 
are said to be equivalent. 
Consider a trivial adjoint-symmetry $Q=\Dop(G)$ of $G=0$,
where $\Dop=\sum_{J} C^J D_J$ is any linear differential operator in total derivatives. 
We have 
\begin{align*}
Q'=\Dop'(G)+\Dop G',
\quad
R_Q = G'{}^*\Dop ,
\end{align*}
with $\Dop'= \sum_{J} (C^J)' D_J$ denoting the Frechet derivative. 
This yields an operator 
\begin{equation}\label{trivS.op}
\S= \Dop'(G)+(\Dop+\Dop^*)G' . 
\end{equation}
Note that, on $\Esp$, the first term $\Dop'(G)$ is trivially zero 
while the second term vanishes when it is applied to any symmetry characteristic $P$,
since $G'(P)|_\Esp =0$. 

Hence, we will regard any two operators \eqref{S.symm-to-adjsymm} 
as being equivalent if they differ by a trivial operator of the form \eqref{trivS.op}. 

It is now interesting to ask when an adjoint-symmetry produced by the operator \eqref{S.symm-to-adjsymm} 
is a multiplier for a conservation law. 

\begin{proposition}\label{prop:condition.Qismultiplier}
For a symmetry $\hat\X_P=P\partial_u$ of $G=0$, 
the adjoint-symmetry $\S(P)$ is a multiplier 
yielding a conservation law of $G=0$ if and only if 
\begin{align}\label{S.multiplier.condition}
(\pr\hat\X_P\S^*+\S^*R_{P}+P'{}^*S^*)G=0
\end{align}
holds off of $\Esp$. 
\end{proposition}

\begin{proof}
$\S(P)$ is a multiplier if and if it satisfies 
\begin{align}\label{EL.SP}
0= E_u(\S(P)G)=G'{}^*(\S(P))+\S(P)'{}^*(G)
\end{align}
off of $\Esp$. 
We proceed to simplify the r.h.s.\ of \eqref{EL.SP}
by writing $\S^*=\sum_J W^J D_J$. 

First, 
from equations \eqref{S*G} to \eqref{rhs.S*G} 
in the proof of Theorem~\ref{thm:op.symm-to-adjsymm}, 
we have 
\begin{align*}
\S^*(G'(P))+ \sum_J (W^J)'(P) D_JG= G'{}^*(\S(P))+ \sum_J (W^J)'{}^*(PD_JG) , 
\end{align*}
which can be rewritten as
\begin{equation}\label{rhs.term1}
G'{}^*(\S(P))=\S^*(G'(P))+\pr\X_{P}(\S^*)(G)  -\sum_J (W^J)'{}^*(PD_JG) . 
\end{equation}
Next, we note 
\begin{align*}
\S(P)'
=\big( \sum_J(-1)^{|J|}D_J(W^JP)\big)'
&=\sum_J(-1)^{|J|}D_J(W^JP'+P(W^J)')\\
&=\S P'+\sum_J(-1)^{|J|}D_JP(W^J)' , 
\end{align*}
and so 
\begin{align*}
\S(P)'{}^*
=P'{}^*\S^*+\sum_J (W^J)'{}^*(P)D_J , 
\end{align*}
which yields
\begin{equation}\label{rhs.term2}
\S(P)'{}^*(G)
=P'{}^*(\S^*(G)) +\sum_J (W^J)'{}^*(P) D_JG . 
\end{equation}
Plugging expressions \eqref{rhs.term1} and \eqref{rhs.term2} 
into the r.h.s.\ of \eqref{EL.SP}, 
we obtain 
\begin{align*}
0 =\S^*(G'(P))+P'{}^*\S^*(G) + \pr\X_{P}(\S^*)(G) . 
\end{align*}
Then plugging in $G'(P) =R_P(G)$ yields \eqref{S.multiplier.condition}. 
\end{proof}

\subsection{Evolution equations and pre-symplectic operators}

To explore the content of 
Proposition~\ref{prop:condition.Qismultiplier} and Theorem~\ref{thm:op.symm-to-adjsymm} further, 
it will be useful to consider the case when $G=0$ is an evolution PDE:
\begin{equation}\label{evolPDE}
G(x,u^{(N)})= u_t - g(x,u^{(N)}) =0
\end{equation}
where $t$ is the time variable and $x$ here denotes the spatial variables. 

Symmetries $\hat\X_P = P\partial_u$ of an evolution equation \eqref{evolPDE}
are determined by 
\begin{equation}\label{evolPDE:symm}
(D_t P - g'(P))|_\Esp =0 . 
\end{equation}
As is well known, 
off of the solution space $\Esp$, 
$P$ obeys 
\begin{equation}
g'(P) - P'(g) =P_t , 
\end{equation}
which implies 
\begin{equation}\label{evolPDE:R_P}
R_P = P' . 
\end{equation}

Adjoint-symmetries $Q$ of an evolution equation \eqref{evolPDE}
are determined by \label{evolPDE:adjsymm}
\begin{equation}
(D_t Q + g'{}^*(Q))|_\Esp =0 , 
\end{equation}
which is the adjoint of the symmetry equation \eqref{evolPDE:symm}. 
Off of the solution space $\Esp$, 
$Q$ similarly obeys 
\begin{equation}
Q'(g) + g'{}^*(Q) = -Q_t ,
\end{equation}
which implies 
\begin{equation}\label{evolPDE:R_Q}
R_Q = -Q'.
\end{equation}
It is well known that an adjoint-symmetry $Q$ is a multiplier 
for a conservation law of $u_t = g$ 
if and only if it satisfies 
\begin{equation}\label{evolPDE:multiplier}
Q' = Q'{}^*
\end{equation}
off of $\Esp$. 
This operator condition is equivalent to the Helmholtz conditions 
which correspond to $Q = E_u(H)$ being a variational derivative of 
some differential function $H$. 

Note that, without loss of generality, 
we can eliminate $t$-derivatives of $u$ from symmetries and adjoint-symmetries
by adding suitable trivial terms to $P$ and $Q$. 

We remark that adjoint-symmetries are sometimes called cosymmetries 
in the literature on integrable systems. 
Strictly speaking, in that context, 
a cosymmetry is a multiplier and thus has a gradient form $Q = E_u(H)$ 
(which is often expressed by writing $E_u =\delta/\delta u$). 

Now, using expression \eqref{evolPDE:R_Q} to get $R_Q^* = -Q'{}^*$, 
we obtain the following results on the operator \eqref{S.symm-to-adjsymm}. 

\begin{theorem}\label{thm:op.presymplectic}
Suppose $u_t=g$ possesses an adjoint-symmetry $Q$ that is not a multiplier. 
Then the linear differential operator \eqref{S.symm-to-adjsymm}, 
under which symmetries are mapped into adjoint-symmetries, 
has the skew-symmetric form 
\begin{align}\label{S.presymplectic}
\S = Q' - Q'{}^* = -S^* . 
\end{align}
For a symmetry $\hat\X_P=P\partial_u$, 
the adjoint-symmetry $\S(P)$ is a multiplier, yielding a conservation law, 
if and only if 
\begin{align}\label{S.presymplectic.multiplier.condition}
\pr\hat\X_P\S+\S P'+P'{}^* \S =0 .
\end{align}
\end{theorem}

Note that the operator $\S$ here is well defined off of $\Esp$, 
since the elimination of $t$-derivatives of $u$ in $Q$ and $P$
removes all gauge freedom in both $\S$ and $S(P)$. 

Since $\S$ is skew-symmetric, it can be viewed as defining a pre-symplectic operator,
namely, a skew, linear operator that maps symmetries to adjoint-symmetries. 
Moreover, when an inverse $\S^{-1}$ exists, 
it can be viewed as defining a pre-Hamiltonian operator \cite{KerKraVer2004}, 
namely, a skew, linear operator that maps adjoint-symmetries to symmetries. 

Finally, 
we remark that the preceding results for evolution equations 
have an elegant formulation 
using the geometrical Lie derivative operator $\lieder{t}$
associated to the flow defined by $u_t=g$. 
In particular, as discussed in \Ref{Anc-review,AncWan}, 
symmetries and adjoint-symmetries are respectively given by 
$\lieder{t}(P\partial_u) =0$
and 
$\lieder{t}(Q\d u) =0$, 
where $P\partial_u$ is a symmetry vector field 
and $Q\d u =0$ is an adjoint-symmetry 1-form. 
The condition for $\S(P)$ to be a multiplier is then given by $\lieder{P}\S=0$,
where $\lieder{P}$ is the Lie derivative with respect to $P\partial_u$.

\subsection{Euler--Lagrange equations and recursion operators}

When a PDE $G=0$ is an Euler--Lagrange equation, 
namely $G=E_u(L)$ for some Lagrangian $L$, 
adjoint-symmetries coincide with symmetries. 
In this situation, 
the generalized pre-symplectic operator  \eqref{S.symm-to-adjsymm} becomes 
a recursion operator for symmetries. 

\begin{theorem}\label{thm:op.recursion}
Suppose an Euler--Lagrange equation $G=0$ 
possesses a symmetry $\hat\X_P=P\partial_u$ (in characteristic form) that is not variational, namely
\begin{align*}
G'(P)=R_P(P),
\quad
E_u(PG)\not\equiv 0
\end{align*}
holds off of $\Esp$,
where $R_P$ is non-zero linear differential operator in total derivatives whose coefficients are non-singular differential functions on $\Esp$. 
Then 
\begin{equation}\label{S.symm.op}
\S := R_P^* + P'
\end{equation}
is a linear differential operator that maps symmetries to symmetries
(namely, a recursion operator). 
\end{theorem}

It is natural to ask under what conditions does this symmetry recursion operator \eqref{S.symm.op}
act as a recursion operator on variational symmetries. 

\begin{proposition}\label{prop:variationalsymm.condition}
If $P\partial_u$ is a variational symmetry of the Euler--Lagrange equation $G=0$,
then $\S(P)\partial_u$ is also a variational symmetry of $G=0$ 
if and only if
\begin{align}\label{EL.S.multiplier.condition}
[\S^*,P'{}^*]G = (\pr\X_{P}\S^*)G . 
\end{align}
\end{proposition}

\begin{proof}
When $P$ is the characteristic of a variational symmetry, 
it satisfies $0=E_u(PG) = G'(P)+P'{}^*(G)$, 
which yields $R_{P}=-P'{}^*$.
Hence the condition \eqref{S.multiplier.condition} becomes
\begin{align*}
(\pr\X_{P}(\S^*)-\S^*P'{}^*+P'{}^*\S^*)G=0 . 
\end{align*}
\end{proof}

\begin{remark}
Since variational symmetries of an Euler--Lagrange equation 
are equivalent to conservation law multipliers by Noether's theorem,  
the condition \eqref{EL.S.multiplier.condition} is necessary 
for the operator \eqref{S.symm.op} to be a recursion operator 
on conservation law multipliers. 
The condition is sufficient if it holds with $\S(P)$ in place of $P$.
\end{remark}

The formula \eqref{S.symm.op} uses symmetries $P$ 
which are lifted off of the solution space $\Esp$ of $G=0$. 
Recall, a symmetry is trivial if it vanishes on $\Esp$,
and two symmetries that differ by a trivial symmetry 
are said to be equivalent. 
Consider a trivial symmetry $P=\Dop(G)$ of $G=0$, 
where $\Dop=\sum_{J} C^J D_J$ is any linear differential operator in total derivatives. 
Similarly to the situation for adjoint-symmetries, 
the resulting recursion operator \eqref{S.symm.op} has the form \eqref{trivS.op}
which vanishes when it is 
applied to any symmetry characteristic $P$ and evaluated on $\Esp$. 

Hence, we will regard any two recursion operators \eqref{S.symm.op} 
as being equivalent if they differ by a trivial operator of the form \eqref{trivS.op}. 

To compare this recursion operator to the formula in Proposition~\ref{prop:EL.recursion},
we consider the situation when $P$ is the characteristic function of a Lie point symmetry. 

\begin{corollary}\label{corr:recursion-pointsymm}
If $P= \eta(x,u) - \xi^i(x,u) u_{i}$ is the characteristic function of 
a non-variational Lie point symmetry of an Euler--Lagrange equation $G=0$, 
then the corresponding symmetry recursion operator \eqref{S.symm.op} is 
$\S = W$,
where $W$ is the variational integrating factor \eqref{W}. 
\end{corollary}

\begin{proof}
We have $R_P = f_\p - \xi^iD_i$ from relation \eqref{X.on.G}. 
Hence, $R_P^* = f_\p + \xi^i_u u_{i} + \xi^i_{x^i} +\xi^iD_i$. 
Since $P'= (\eta - \xi^i u_{i})' = \eta_u -u_{i}\xi^i_u -\xi^i D_i$, 
we get 
$R_P^* + P' = f_\p + \xi^i_{x^i} +\eta_u =W$ after cancellations. 
\end{proof}

Consequently, 
finding genuine (non-multiplicative) recursion operators \eqref{S.symm.op} 
requires going beyond Lie point symmetries. 

We first consider contact symmetries. 
Recall that, in the present setting of a single dependent variable, 
a contact symmetry is equivalent to a first-order symmetry in characteristic form. 
In particular, 
$\hat\X_P=P(x,u^{(1)})\partial_u$ 
corresponds to $\X = \xi^i(x,u^{(1)})\partial_{x^i} + \eta(x,u^{(1)})\partial_u$
where $\xi^i= -P_{u_i}$ and $\eta= P -u_i P_{u_i}$. 

\begin{corollary}\label{corr:recursion.contactsymm}
If $P(x,u^{(1)})$ is the characteristic function of a non-variational contact symmetry of 
an Euler--Lagrange equation $G=0$, 
with $\pr\X_P G = f G$, 
then the corresponding symmetry recursion operator \eqref{S.symm.op} is 
\begin{equation}\label{S.contact.symm.op} 
\S = f+P_{u} - D_{i}P_{u_i}
\end{equation}
which is a variational integrating factor. 
\end{corollary}

\begin{proof}
First, note $\pr\X_P(G) = f G$ holds for some differential function $f$, 
because the contact symmetry acts as a generator of a point transformation on $(x,u,\partial u)$
which prolongs to a point transformation in jet space (of any finite order).
Next, we have
$P'=\eta' - (u_i\xi^i)' = \eta_{u} + \eta_{u_i} D_i -u_i(\xi^i_{u} + \xi^i_{u_j}D_j) -\xi^iD_i$, 
and 
$R_P= f - \xi^i D_i$
so thus
$R_P^*= f  + (D_i\xi^i) + \xi^iD_i$ 
where $D_i\xi^i = \xi^i_{x^i} + u_i\xi^i_{u} + u_{ij}\xi^i_{u_j}$. 
Substitution of the expressions for $\xi^i$ and $\eta$ in terms of $P$ 
then leads to the result \eqref{S.contact.symm.op}. 
\end{proof}

Therefore, the formula \eqref{S.symm.op} can yield a genuine recursion operator 
only when $\hat\X_P=P\partial_u$ is either a higher symmetry or a nonlocal symmetry of $G=0$.

\section{Examples}\label{sec:examples}

We will first give two examples to illustrate for Euler--Lagrange PDEs 
how Theorem~\ref{thm:op.recursion} can be used to extract a symmetry recursion operator from a higher symmetry. 
In these examples we also will illustrate the condition in Proposition~\ref{prop:variationalsymm.condition} 
for this operator to be a variational-symmetry recursion operator. 

Next we will give two examples to illustrate how Theorem~\ref{thm:op.presymplectic}
yields a pre-symplectic operator from an adjoint-symmetry 
for nonlinear evolution equations. 

Finally, we will give two examples of Theorem~\ref{thm:op.symm-to-adjsymm} 
for obtaining a generalized pre-symplectic operator from an adjoint-symmetry of 
a nonlinear PDE which is not of evolution form or Euler--Lagrange form. 
The condition in Proposition~\ref{prop:condition.Qismultiplier} 
for this operator to map symmetries into conservation law multipliers will also be illustrated.

\subsection{Euler--Lagrange PDE examples}\hfil

\emph{Linear Wave Equation}: 
The first example is the wave equation 
\begin{equation}\label{linearwaveeqn}
u_{tt} = a u_{xx} + b u_{t} + c u
\end{equation}
with constant coefficients $a>0,b,c$. 
This equation becomes Euler--Lagrange after multiplication by $e^{-bt}$ 
which is a variational integrating factor. 
In particular, 
\begin{equation}
G = e^{-bt}(u_{tt} - a u_{xx} - b u_{t} - c u)=E_u(L)
\end{equation}
holds for $L = \tfrac{1}{2}e^{-bt}(-u_t{}^2  + a u_x{}^2 - cu^2)$. 

The scaling symmetry $\X = u\partial_u$ of $G=0$ is non-variational, 
and so is the time translation symmetry when $b\neq 0$. 
Because this wave equation is linear and translation invariant, 
it admits $D_t$ and $D_x$ as symmetry recursion operators. 
Applying even powers of these operators to the scaling symmetry yields 
higher symmetries $\X_{(p,q)} = u_{pt\,qx}\partial_u$ given by $p+q=2,4,\ldots$, 
with $p\geq 0$ and $q\geq 0$ denoting the number of $t$ and $x$ derivatives, respectively. 
Each of these higher symmetries can be shown to be non-variational. 
For instance, in the case when $p=0$ and $q=2,4,\ldots$, 
it is easy to see that 
\begin{equation*}
E_u(u_{qx}G) = 2 e^{-bt}(u_{2t\,qx} - a u_{(q+2)x} - b u_{t\,qx} - c u_{qx}) \not\equiv 0 . 
\end{equation*}
A similar result holds in the general case. 
Similarly, $\X_{(p,q)}$ can be shown to be variational when $p+q$ is odd. 

These higher symmetries of the linear wave equation 
have $R_{(p,q)} = P_{(p,q)}'= D_{t\mathstrut}^p D_{x\mathstrut}^q$,
and hence $R_{(p,q)}^* = R_{(p,q)}$ holds for the non-variational symmetries 
since $p+q$ is even. 
Thus, the resulting symmetry recursion operators \eqref{S.symm.op} 
coming from the non-variational symmetries are given by 
$\S_{(p,q)}= 2 D_{t\mathstrut}^p D_{x\mathstrut}^q$
which can be expressed as compositions of the primitive operators
\begin{equation}\label{linearwave.Sop}
\S_{(1,1)}= 2 D_{t\mathstrut}D_{x\mathstrut}, 
\quad
\S_{(2,0)}= 2 D_{t\mathstrut}^2,
\quad
\S_{(0,2)}= 2 D_{x\mathstrut}^2 .
\end{equation}
Note $S_{(p,q)}^*=S_{(p,q)}$ since $p+q$ is even. 

The necessary condition \eqref{EL.S.multiplier.condition} for the operators \eqref{linearwave.Sop} 
to act as recursion operators on variational symmetries 
can be straightforwardly checked to hold for $\X_{(p,q)}$ with $p+q$ being odd.
Specifically, each $\S$ commutes with 
$P_{(p,q)}^{\prime *} = - P_{(p,q)}'= - D_{t\mathstrut}^p D_{x\mathstrut}^q$,
while $\pr\X_{(p,q)}\S$ vanishes,
so that $[\S^*,P_{(p,q)}^{\prime *}]= \pr\X_{(p,q)}\S^*=0$.  
Moreover, these properties are also sufficient to show that each operator $\S$ 
is a recursion operator on the variational symmetries $\X_{(p,q)}$ when $p+q$ is odd.

\emph{Korteweg--de Vries Equation}: 
The second example is the KdV equation
\begin{equation}\label{kdveqn}
u_t + u u_x +  u_{xxx} =0 , 
\end{equation}
which is an integrable nonlinear PDE having a symmetry recursion operator \cite{AblCla-book,Olv-book}
\begin{equation}\label{kdv.recursionop}
\Rop = D_x^2 + \tfrac{1}{3} u + \tfrac{1}{3} D_x u D_x^{-1} . 
\end{equation}
This PDE becomes an Euler--Lagrange equation through $u=v_x$, 
where $v$ is a potential yielding
\begin{equation}
G = v_{tx} + v_{x} v_{xx} +v_{xxxx}=E_v(L)
\end{equation}
for $L = -\tfrac{1}{2} v_tv_x - \tfrac{1}{6}v_x{}^3 + \tfrac{1}{2}v_{xx}{}^2$. 
The KdV recursion operator is transformed into a corresponding recursion operator for $G=0$:
\begin{equation}\label{pkdv.recursionop}
\widetilde\Rop = D_x^{-1} \Rop D_x = D_x^2 + \tfrac{1}{3} v_x + \tfrac{1}{3} D_x^{-1} v_x D_x . 
\end{equation}
This operator satisfies $\widetilde\Rop^*=\widetilde\Rop$,
which corresponds to $G'{}^*=G'$. 
It can be applied to any symmetry that is inherited from the underlying evolution equation
$v_t + \tfrac{1}{2}v_x{}^2 +v_{xxx}=0$. 

The scaling symmetry $\X = 3t\partial_t +x\partial_x -v\partial_v$ is non-variational. 
A higher scaling-type symmetry can be obtained by applying the recursion operator \eqref{pkdv.recursionop}
to the characteristic function 
\begin{equation}
P = -(v+x v_x + 3t v_t)
\end{equation}
of the scaling symmetry.
This yields
\begin{equation}\label{kdv.P1}
P_1 := \widetilde\Rop(P) = t(3v_{xxxxx} +5v_{x}v_{xxx}+\tfrac{5}{2}v_{xx}^2+\tfrac{5}{6}v_{x}^3) -x(v_{xxx} + \tfrac{1}{2}v_{x}^2) -3v_{xx}- \tfrac{1}{3}vv_{x} -\tfrac{1}{2}D_x^{-1}(v_x^2) , 
\end{equation}
giving a symmetry $\hat\X_1 = P_1\partial_v$ of the 
KdV equation in potential form
\begin{equation}\label{pkdveqn}
v_{tx} + v_{x} v_{xx} +v_{xxxx} =0 .
\end{equation}
Then we have, by direct computation, 
\begin{equation}
\begin{aligned}
R_{P_1} & = 
3tD_x^5
+(5tv_{x}-x)D_x^3
+(10tv_{xx}-4)D_x^2
+(10tv_{xxx}+\tfrac{5}{2}tv_{x}^2-xv_{x}-\tfrac{1}{3}v)D_x
\\&\qquad
+5tv_{xxxx} +5tv_{x}v_{xx} -xv_{xx}-\tfrac{8}{3}v_{x}
-\tfrac{1}{3}v_{xx}D_x^{-1}
\end{aligned}
\end{equation}
and 
\begin{equation}
R_{P_1}^* = 
-3tD_x^5
-(5tv_{x}-x)D_x^3
-(5tv_{xx}+1)D_x^2
-(5tv_{xxx} +\tfrac{5}{2}tv_{x}^2-xv_{x}-\tfrac{1}{3}v)D_x
-v_{x}
-\tfrac{1}{3}D_x^{-1}(v_{x}D_x)
\end{equation}
as well as 
\begin{equation}
P_1' = 
3tD_x^5
+(5tv_{x}-x)D_x^3
+(5tv_{xx} -3)D_x^2
+(5tv_{xxx}+\tfrac{5}{2}tv_{x}^2-xv_{x}-\tfrac{1}{3}v )D_x
-\tfrac{1}{3}v_{x}
-D_x^{-1} (v_{x}D_x) .
\end{equation}
The resulting symmetry recursion operator \eqref{S.symm.op} is given by 
\begin{equation}\label{kdv.Sop.recusion}
\S = -4\widetilde\Rop .
\end{equation}

The variational Lie point symmetries of the KdV equation in potential form \eqref{pkdveqn}
are well known to consist of 
time translation, space translation, and a Galilean boost,
given by the characteristic functions $-v_t$, $-v_x$, $x-t v_x$,
as well as shifts whose characteristic function is $f(t)$ for an arbitrary function. 
When $\widetilde\Rop$ is applied to the boost symmetry, 
it yields the scaling symmetry. 
(Note $\widetilde\Rop$ annihilates a shift, $f(t)$.)

The well-known hierarchy of variational higher symmetries in potential form 
is obtained from powers of $\widetilde\Rop$ applied to the space translation symmetry:
\begin{equation}\label{kdv.varsymms}
P_{(k)} = \widetilde\Rop^k(-v_x), 
\quad
k=0,1,2,\ldots
\end{equation}
where $P_{(0)}=-v_x$ is the characteristic of the space translation, 
and $P_{(1)}=v_t$ corresponds to the time translation. 

Consequently, 
since $\S$ is a recursion operator on the variational symmetries $P_{(k)}\partial_u$,
it must satisfy the condition \eqref{EL.S.multiplier.condition} 
for all $k\geq 0$. 
To see the content of this condition, 
consider $P_{(0)}=-v_x$. 
Since $P_{(0)}'=-D_x = -P_{(0)}^{\prime *}$, 
we see that the condition \eqref{EL.S.multiplier.condition} becomes
\begin{equation}
\pr\X_{u_x}\S + [\S,D_x]=0 . 
\end{equation}
This is the well-known property \cite{Olv-book} that 
the recursion operator \eqref{pkdv.recursionop} is invariant under the flow defined by $u_t=u_x$.

\subsection{Evolution PDE examples}\hfil

\emph{Airy Equation}: 
The third-order dispersive linear PDE 
\begin{equation}\label{airyeqn}
u_{t} + a(x) u_{xxx} =0
\end{equation}
is an Airy equation with a non-constant coefficient $a(x)$. 
It has no variational structure. 
Its adjoint-symmetries of the form $Q(t,x,u,u_t,u_x)$ 
are readily found by solving the determining equation 
$(D_t Q + D_x^3(a(x)Q))|_\Esp =0$. 
When $a(x)$ is arbitrary, 
the determining equation yields, other than solutions of the Airy equation itself, 
a linear combination of $u/a(x)$ and $u_t/a(x)$. 

The first of these adjoint-symmetries turns out to be a multiplier for a conservation law
\begin{equation*}
D_t( \tfrac{1}{2} u^2/a(x) ) + D_x( u u_{xx} -\tfrac{1}{2}u_x{}^2 )
= (u_{t} + a(x) u_{xxx}) u/a(x) =0
\end{equation*}
for solutions $u(t,x)$. 
The other adjoint-symmetry 
\begin{equation}\label{airy.Q}
Q=u_t/a(x)
\end{equation}
is not a multiplier,
since $E_u(G u_t/a(x)) = -2D_tG /a(x) \not\equiv 0$ where $G=u_{t} + a(x) u_{xxx}$. 

By a simple computation, we have 
$Q'= (1/a(x))D_t$ 
and 
$G'(Q) = -D_tG /a(x)$ 
which yields $R_Q = -(1/a(x))D_t$ and hence $R_Q^* = (1/a(x))D_t$. 
Thus, the resulting pre-symplectic operator \eqref{S.presymplectic} is given by 
\begin{equation}\label{airy.Sop}
\S= (2/a(x))D_t = -\S^* . 
\end{equation}

Because this Airy equation is linear and time-translation invariant, 
it admits $D_t$ as a recursion operator on symmetries and adjoint-symmetries. 
In particular, 
apart from elementary symmetries
whose characteristic function is given by solutions of the Airy equation itself, 
all other symmetries can be shown to be generated from 
powers of $D_t$ applied to the characteristic function of the scaling symmetry $u\partial_u$. 
Thus, $\X_{(k)} := u_{kt}\partial_u$, 
with $k\geq 0$ denoting the number of $t$ derivatives, 
comprises a sequence of higher symmetries of the Airy equation \eqref{airyeqn}. 
The pre-symplectic operator \eqref{airy.Sop} produces a corresponding sequence of 
adjoint-symmetries:
$Q_{(k)}= \S(u_{kt}) = (2/a(x)) u_{k+1\,t}$, $k=0,1,2,\ldots$,
where $Q_{(0)}=2Q$ is the adjoint-symmetry \eqref{airy.Q}. 

The condition \eqref{S.presymplectic.multiplier.condition} 
for $\S$ to produce a multiplier from $P_{(k)}=u_{kt}$ 
can be readily seen to imply that $k$ is odd,
since we have 
$\pr\hat\X_{P_{(k)}}(\S) =0$, 
$P_{(k)}'=D_t^k$ 
and $P_{(k)}^{\prime *}=(-1)^k D_t^k$,
whereby 
\begin{equation*}
0 = \pr\hat\X_{P_{(k)}}\S +\S P_{(k)}' +P_{(k)}^{\prime *}\S = (2(1+(-1)^k)/a(x))D_t^{k+1}
\end{equation*}
yields $(-1)^k = -1$. 
This condition agrees with the multiplier determining equation 
\begin{equation*}
0 =E_u(G Q_{(k)}) = -2(1+(-1)^k)D_t^{k+1}G /a(x) . 
\end{equation*}

\emph{Generalized Korteweg--de Vries Equation}: 
The gKdV equation 
\begin{equation}\label{gkdveqn}
u_t + u^p u_x +  u_{xxx} =0
\end{equation}
is a generalization of the KdV equation involving a nonlinearity power $p\neq 0$. 
For all powers, 
this equation has a Hamiltonian formulation 
$u_t = -\Dop(\delta H/\delta u)$ 
where $H=\int \tfrac{1}{(p+1)(p+2)} u^{p+2} - \tfrac{1}{2} u_x^2\,dx$
is the Hamiltonian functional,
and $\Dop=D_x$ is a Hamiltonian operator \cite{Olv-book}. 

The gKdV equation is not Euler--Lagrange as it stands. 
It has a scaling symmetry $\X = 3t\partial_t + x\partial_x -\tfrac{2}{p} u\partial_u$. 
The corresponding characteristic form of this symmetry is given by 
\begin{equation}
P= -(\tfrac{2}{p} u + x u_x + 3tu_t), 
\end{equation}
which can be expressed equivalently as 
\begin{equation}
P|_\Esp = -(\tfrac{2}{p} u + x u_x - 3t(u^pu_x +u_{xxx}))
\end{equation}
on the solution space. 

Since $D_x$ is a Hamiltonian operator, it maps adjoint-symmetries into symmetries. 
Hence, its inverse $D_x^{-1}$ maps symmetries into adjoint-symmetries. 
Applying this latter operator to the scaling symmetry, 
we obtain the adjoint-symmetry 
\begin{equation}\label{gkdv.Q}
Q = (1-\tfrac{2}{p}) v -x u  + 3t(\tfrac{1}{p+1}u^{p+1} +u_{xx})
\end{equation}
where $v$ is a potential defined by $u=v_x$,
and where $v_t = D_x^{-1} u_t$ is expressed through the gKdV equation \eqref{gkdveqn}. 
In particular, this adjoint-symmetry is nonlocal. 
By direct computation, we have 
\begin{equation}
Q' = 3tD_x^2  +3t u^p -x  +(1-\tfrac{2}{p}) D_x^{-1}
\end{equation}
and 
\begin{equation}
R_Q = - 3tD_x^2  -3t u^p + x -(1-\tfrac{2}{p}) D_x^{-1} , 
\end{equation}
whence 
\begin{equation}
R_Q^* = -3tD_x^2  -3t u^p + x  + (1-\tfrac{2}{p}) D_x^{-1} . 
\end{equation}
This yields 
\begin{equation}\label{gkdv-Sop}
\S = 2(1-\tfrac{2}{p}) D_x^{-1} , 
\end{equation}
which is a multiple of the inverse of the gKdV Hamiltonian operator $\Dop=D_x$. 

Note this operator $\S$ is trivial in the case $p=2$ which corresponds to the modified KdV equation. 
This is precisely the case in which $Q$ is a multiplier for a conservation law
and corresponds to the scaling symmetry being a variational symmetry 
when the mKdV equation is written in terms of the potential $v$. 

Finally, we will tie this example and the previous KdV example together to show how 
two Hamiltonian operators can be derived for the KdV equation 
from its scaling symmetry. 

For $p=1$, 
the gKdV equation \eqref{gkdveqn} coincides with the KdV equation \eqref{kdveqn}. 
The scaling adjoint-symmetry \eqref{gkdv.Q} becomes 
\begin{equation}\label{kdv.Q}
Q = -v -x u  + 3t(\tfrac{1}{2}u^2 +u_{xx}) , 
\end{equation}
which produces the operator 
\begin{equation}
\S = -2 D_x^{-1} . 
\end{equation}
The inverse of this operator is a multiple of the first Hamiltonian operator $\Dop=D_x$ 
of the KdV equation. 
In particular, 
$u_t = -\Dop(\delta H/\delta u)$ 
where $H=\int \tfrac{1}{6} u^3 - \tfrac{1}{2} u_x^2\,dx$
is the Hamiltonian functional. 

A recursion operator on adjoint-symmetries of the KdV equation is given by 
the adjoint of the symmetry recursion operator \eqref{kdv.recursionop}, 
\begin{equation}\label{kdv.adjrecursionop}
\Rop^* = D_x^2 + \tfrac{1}{3} u + \tfrac{1}{3} D_x^{-1} u D_x .
\end{equation}
Applying this operator to the adjoint-symmetry \eqref{kdv.Q}, 
we obtain a higher adjoint-symmetry 
\begin{equation}\label{kdv.higherQ}
Q_1: = \Rop^* Q 
=t(3u_{xxxx}+5uu_{xx}+\tfrac{5}{2}u_{x}^2 +\tfrac{5}{6}u^3)
-x(u_{xx} + \tfrac{1}{2}u^2)
-3u_{x} -\tfrac{1}{3}uv -\tfrac{1}{2}D_x^{-1}u^2
\end{equation}
which coincides with the higher symmetry \eqref{kdv.P1} in potential form, $P_1 = Q_1$. 
Then, by direct computation, 
we have 
\begin{equation}
Q_1' = 
-D_x^{-1}u
-\tfrac{1}{3}uD_x^{-1}
+(5u_{xx}+\tfrac{5}{2}u^2)t -xu -\tfrac{1}{3}v
+(5tu_{x}-3)D_x
+(5tu-x)D_x^2
+3tD_x^4
\end{equation}
and 
\begin{equation}
R_{Q_1} = 
D_x^{-1}u
+ \tfrac{1}{3}uD_x^{-1} 
+(-5tu_{xx}-\tfrac{5}{2}tu^2+xu+\tfrac{1}{3}v)
+(-5tu_{x}+3)D_x
+(-5tu+x) D_x^2
-3tD_x^4 , 
\end{equation}
whence 
\begin{equation}
R_{Q_1}^* = 
-uD_x^{-1}
-\tfrac{1}{3} D_x^{-1}u
-5tu_{xx}-\tfrac{5}{2}tu^2+xu+\tfrac{1}{3}v
-(5tu_{x}+1)D_x
-(5tu-x)D_x^2
-3tD_x^4 .
\end{equation}
Therefore, we obtain 
\begin{equation}\label{kdv.Sop}
\S = -4( D_x + \tfrac{1}{3} (u D_x^{-1} + D_x^{-1}u) ) 
= -4 D_x^{-1}\Rop
= -4 \Rop^* D_x^{-1} ,
\end{equation}
which is an operator that maps symmetries into adjoint-symmetries, 
where $\Rop$ is the symmetry recursion operator \eqref{kdv.recursionop}
and $\Rop^*$ is its adjoint. 

The formal inverse of this operator \eqref{kdv.Sop} is given by 
\begin{equation}\label{kdv.invSop}
\S^{-1} = -\tfrac{1}{4} \Rop^{-1}D_x = -\tfrac{1}{4} D_x(\Rop^*)^{-1} ,
\end{equation}
which is a nonlocal operator that maps adjoint-symmetries into symmetries. 
It constitutes a second Hamiltonian operator for the KdV equation. 
In particular, 
$u_t = -\Hop(\delta E/\delta u)$ 
where $E=\int \tfrac{5}{72} u^4 -\tfrac{5}{6} u u_{x}^2 + \tfrac{1}{2} u_{xx}^2\,dx$
is the Hamiltonian functional,
and $\Hop=\Rop^{-1}D_x= D_x(\Rop^*)^{-1}$ is a Hamiltonian operator.

\subsection{Non-evolution PDE examples}\hfil

\emph{Peakon Equation}: 
Every nonlinear dispersive wave equation
\begin{equation*}
u_{t}-u_{txx} + f(u,u_x)(u-u_{xx}) + (g(u,u_x)(u-u_{xx})){}_x =0
\end{equation*}
possesses multi-peakon solutions \cite{AncRec2019} 
which are given by a linear superposition of peaked waves 
$A(t)\exp(-|x-X(t)|)$ having a time-dependent amplitude $A(t)$ and position $X(t)$. 
This class of wave equations includes the integrable Camassa--Holm (CH) equation \cite{CamHol,CamHolHym} 
\begin{equation*}
m_t + u_x m + (um){}_x =0, 
\quad
m=u-u_{xx},
\end{equation*}
which arises from a shallow water approximation 
for the Eulerian equations of incompressible fluid flow. 
The CH equation is related (via a reciprocal transformation) to a negative flow 
in the KdV hierarchy of integrable equations.
A modified version of the CH equation is the mCH equation \cite{Fok95a,Olv-Ros,Fok-Olv-Ros,Fuc,Qia-Li}
(also known as the FORQ equation)
\begin{equation*}
m_t + ((u^2- u_x{}^2)m){}_x =0,
\quad
m=u-u_{xx} .
\end{equation*}
This is an integrable equation which is similarly related to a negative flow 
in the mKdV hierarchy of integrable equations.

There is a nonlinear unified generalization of the CH and mCH equations \cite{AncRec2019} 
\begin{equation}\label{gchmcheqn}
m_t +a u_x (u^2 - u_x{}^2)^{k/2} m + (a u(u^2 - u_x{}^2)^{k/2} m + b (u^2 -u_x{}^2)^{(k+1)/2} m){}_x = 0,
\quad
m=u-u_{xx}
\end{equation}
with an arbitrary nonlinearity power $k\geq 0$,
and constant coefficients $a,b$. 
In particular, the CH equation is given by $k=2$, $b=0$;
and the mCH equation is given by $k=1$, $a=0$.
Those two equations share a Hamiltonian structure that is retained for the generalized equation \eqref{gchmcheqn}: 
$m_t = \Delta D_x(\delta H/\delta m)$,
where $\Delta = 1-D_x^2$. 

With respect to $u$, 
the generalized equation \eqref{gchmcheqn} is neither an evolution equation nor an Euler--Lagrange equation. 
Its adjoint-symmetries $Q$ are determined by 
\begin{equation}
\begin{aligned}
& 
\big( - \Delta D_t Q
-b \big( (u^2-u_x{}^2)^{(k+1)/2} \Delta D_xQ +(k+1)(u^2-u_x{}^2)^{(k-1)/2} m (u D_x Q -u_{x}D_x^2 Q) \big)
\\&\qquad
-a\big( (u^2-u_x{}^2)^{k/2} (uD_x Q -u_{x}D_x^2 Q+(k+1)m D_x Q+u \Delta D_x Q) 
\\&\qquad
-k u_{x}(u^2-u_x{}^2)^{(k-2)/2} m (uD_x^2 Q -u_{x}D_x Q) \big)
\big)\big|_\Esp =0 .
\end{aligned}
\end{equation}

The generalized equation \eqref{gchmcheqn} has a scaling symmetry 
$\X = (k+1)t \partial_t -u\partial_u$. 
Similarly to the gKdV example, 
a scaling adjoint-symmetry of equation \eqref{gchmcheqn} is given by 
\begin{equation}\label{gchmch.Q}
Q = -v - (k+1)t v_t
\end{equation}
where $v$ is a potential defined by $u=v_x$. 
We then have 
\begin{equation}
Q'=-(k+1)t D_t D_x^{-1} -D_x^{-1}
\end{equation}
and 
\begin{equation}
R_Q=(k+1) t D_x^{-1}D_t +(k+2) D_x^{-1} , 
\end{equation}
whence
\begin{equation}
R_Q^*=(k+1) t D_t D_x^{-1} -D_x^{-1} . 
\end{equation}
Thus, the resulting generalized pre-symplectic operator \eqref{S.symm-to-adjsymm} 
is given by 
\begin{equation}\label{gchmch.Sop}
\S = -2D_x^{-1} . 
\end{equation}
Note this operator is skew. 

In addition to the scaling symmetry, 
the generalized equation \eqref{gchmcheqn} also possesses 
translation symmetries. 
When $\S$ is applied to their characteristics $P=-u_x$ and $P=-u_t$, 
it yields the adjoint symmetries 
\begin{align}
& \S(-u_x) = 2 u,
\\
& \S(-u_t) = 2( u_{tx} - \tfrac{a}{k+2} (u^2-u_x{}^2)^{(k+2)/2} -(a (u^2-u_x{}^2)^{k/2}u+b (u^2-u_x{}^2)^{(k+1)/2})m ).
\end{align}
The condition \eqref{S.multiplier.condition} for these adjoint-symmetries to be multipliers
simplifies to be $(\pr\hat\X_{P}\S+\S R_{P} -P' S)G=0$,
where $G$ is the left side of equation \eqref{gchmcheqn}. 

In the case of the space translation, 
we have $P=-u_x$ and $R_P = P' = -D_x$,
so the condition becomes $[\S,D_x]G=0$,
which holds due to $[\S,D_x]=0$. 
Hence $Q=u$ is a multiplier. 
It yields a conservation law for the $H^1$ norm (up to an overall factor of $\tfrac{1}{2}$) of solutions $u(t,x)$:
$\frac{d}{dt}\int_\Rnum u^2 + u_x{}^2\,dx =0$. 

In the case of the time translation, 
the condition similarly becomes $[\S,D_t]G=0$,
which holds due to $[\S,D_t]=0$. 
Hence $Q=-u_t$ is a multiplier. 
The resulting conservation law is given by 
\begin{equation}
\begin{aligned}
\frac{d}{dt}\int_\Rnum \big(
& 
\tfrac{a}{k+2} u (u^2-u_x{}^2)^{(k+2)/2}
+\tfrac{b}{k+3} (u^2-u_x{}^2)^{(k+3)/2}
\\&
+a u u_x \smallint (u^2-u_x{}^2)^{k/2}\,du_x
+b u_x \smallint (u^2-u_x{}^2)^{(k+1)/2}\,du_x
\big)\,dx =0 .
\end{aligned}
\end{equation}

\emph{2D Boussinesq Equation}: 
The Kadomtsev--Petviashvili (KP) equation \cite{KadPet} is a well-known two-dimensional generalization of the KdV equation. 
An analogous two-dimensional generalization of the Boussinesq equation is given by \cite{AncGanRec2018}
\begin{equation}\label{2Dbousseqn}
u_{tt}- (u_{xx} + u_{yy} +a (u^2){}_{xx} +b u_{xxxx}) =0
\end{equation}
with constant coefficients $a,b\neq 0$. 
This equation \eqref{2Dbousseqn} is neither an evolution equation nor an Euler--Lagrange equation with respect to $u$. 
Its adjoint-symmetries $Q$ are determined by 
\begin{equation}
\big( D_t^2 Q - (D_x^2 Q + D_y^2 Q + 2a u D_x^2 Q + b D_x^4 Q) \big)\big|_\Esp =0 .
\end{equation}

Equation \eqref{2Dbousseqn} does not have a scaling symmetry. 
However, it does possess a combined scaling-shift symmetry \cite{AncGanRec2018}
$\X = 2t \partial_t + x\partial_x +2y\partial_y -2(u+\tfrac{1}{a})\partial_u$.
It also possesses a similar adjoint-symmetry 
\begin{equation}\label{2Dbouss.Q}
Q=\tfrac{1}{a}(y^2 -x^2) -x v_x -2y v_y - 2t v_t 
\end{equation}
where $v$ is a potential defined by $u=v_{xx}$. 
For this nonlocal adjoint-symmetry, 
we have 
\begin{equation}
Q'= -x D_x^{-1} -2 t D_t D_x^{-2} -2 y D_y D_x^{-2}
\end{equation}
and 
\begin{equation}
R_Q=-2t D_t D_x^{-2}  -2y D_y D_x^{-2}  -x D_x^{-1} -4 D_x^{-2} ,
\end{equation}
whence
\begin{equation}
R_Q^* = -D_x^{-2} + xD_x^{-1} +2t D_t D_x^{-1} +2y D_y D_x^{-2} .
\end{equation}
This yields a generalized pre-symplectic operator \eqref{S.symm-to-adjsymm} 
given by 
\begin{equation}\label{2Dbouss.Sop}
\S = -D_x^{-2} .
\end{equation}
Note this operator is symmetric, $\S^*=\S$. 

In addition to the scaling-shift symmetry, 
the 2D Boussinesq equation \eqref{2Dbousseqn} possesses \cite{AncGanRec2018}
a boost symmetry $\X = y\partial_t +t\partial_y$, 
a family of travelling wave symmetries $\X = (f_+(t+y)+f_-(t-y))\partial_u$ 
where $f_\pm(t\pm y)$ are arbitrary functions, 
and time and space translation symmetries. 
When we apply $\S$ to their characteristics, 
we obtain the adjoint-symmetries
\begin{gather}
\S(-u_x) = v_x,
\quad 
\S(-u_y) = v_y,
\quad
\S(-u_t) = v_t,
\label{2Dbouss.translation.Q}
\\
\S(-y u_t -t u_y) = y v_t + t v_y,
\label{2Dbouss.boost.Q}
\\
\S(f_\pm(t\pm y)) = \tfrac{1}{2} x^2 f_\pm(t\pm y) ,
\label{2Dbouss.travelwave.Q}
\end{gather}
and the adjoint-symmetry \eqref{2Dbouss.Q} from the scaling-shift. 

The condition \eqref{S.multiplier.condition} for these adjoint-symmetries \eqref{2Dbouss.translation.Q}--\eqref{2Dbouss.travelwave.Q} 
to be multipliers simplifies to be $(\pr\hat\X_{P}\S+\S R_{P} +P'{}^*\S)G=0$,
where $G$ is the left side of equation \eqref{2Dbousseqn}. 
This condition holds for all them, 
because they have $R_P= P'=-P'{}^*$, $[\S,P']=0$, $\pr\hat\X_{P}\S=0$. 
The resulting conservation laws are known to be \cite{AncGanRec2018}, respectively, 
an energy and two spatial momenta, a boost momentum, 
and analogs of the transverse momentum for the linear wave equation. 

For the scaling-shift adjoint-symmetry \eqref{2Dbouss.Q}, 
we have $R_P = P' -4$ and $P'{}^* = 1-P'$ where $P'=-2-2tD_t-xD_x-2yD_y$, 
and hence the multiplier condition \eqref{S.multiplier.condition} 
reduces to the form $([\S,P'] -3\S)G=0$. 
A simple computation shows that $[\S,P'] = \S$, whereby 
$[\S,P'] -3\S = -2\S \not\equiv 0$. 
Thus, the adjoint-symmetry \eqref{2Dbouss.Q} is not a multiplier, 
in accordance with Theorem~\ref{thm:op.symm-to-adjsymm}.

\section{Concluding Remarks}\label{sec:remarks}

We have developed a new way to derive generalized pre-symplectic operators for PDEs, 
both linear and nonlinear, 
using just adjoint-symmetries. 
As a corollary, this yields recursion operators for Euler--Lagrange PDEs,
using just non-variational symmetries. 

Through examples, 
we see that scaling symmetries can used to obtain the well-known recursion operators 
for the KdV equation and for linear wave equations. 
We also see that scaling-type adjoint-symmetries yield the symplectic operators known 
for the Camassa--Holm and modified Camassa--Holm peakon equations 
as well as for the Boussinesq equation. 
This suggests that there may be a deeper connection between our main formula 
and the theory of master symmetries as based on scaling vector fields. 

All of our results can be extended in a straightforward fashion 
to systems of PDEs with any number of dependent variables. 
For future work, we plan to carry out a classification of 
non-variational symmetries and adjoint-symmetries 
for various interesting classes of PDEs. 
We also plan to investigate how our main formula fits into a more geometrical formulation of recursion operators and pre-symplectic operators.

\section*{Appendix: Index notation}

We adapt the index notation used in \Ref{Olv-book}. 

Partial derivatives of $u$ with respect to $x^i$ are denoted
$u_i = \partial_{x^i} u$ 
and 
$u_{i\ldots j} = \partial_{x^i}\ldots\partial_{x^j}u$. 
Summation over a repeated single index $i$ is assumed unless otherwise noted. 

Higher derivatives are denoted by multi-indices which are defined by 
$u_J:= u_{j_1\ldots j_m}$, 
where $J=\{j_1,\ldots,j_m\}$ is an unordered set that represents the differentiation indices and $|J|= m$ is the differential order. 
Summation over a repeated multi-index will be denoted by $\sum$. 

If $J=\{j_1,\ldots,j_m\}$ and $K=\{k_1,\ldots,k_{m'}\}$ are multi-indices, 
then their union is defined by $J,K := \{j_1,\ldots,j_m,k_1,\ldots,k_{m'}\}$ 
and $|J,K| := m+m'$. 
If $K\subset J$, 
then $J/K$ denotes set of indices remaining in $J$ after all indices in $K$ are removed. 
Note the differential order is $|J/K| = |J|-|K|$. 

In this notation, total derivatives with respect to $x^i$ are given by 
$D_J := D_{j_1}\cdots D_{j_m}$
where 
$D_j = D_{x^j} = \partial_{x^j} + \sum_{K} u_{j,K}\partial_{u_K}$. 

For expressing the product rule $D_J(fg)$ in a short notation, 
it is useful to introduce the binomial coefficient 
$\smallbinom{J}{K} := \frac{(\# J)!}{(\# K)!(\# J/K)!}$. 
Here 
$\# J\!:= (\#1,\ldots,\# n)$ 
where $\# i$ denotes the multiplicity (number of occurrences) of each integer 
$i=1,\ldots,n$ in the set $J$;
and $(\# J)! := \prod_{1\leq i\leq n} (\#i)!$ where $!$ is the standard factorial. 
A useful identity is $\smallbinom{J}{K} = \smallbinom{J}{J/K}$. 

Then $D_J(fg) = \sum_{K\subseteq J}\smallbinom{J}{K} D_{J/K} f D_K g$ is the product rule.

\section*{Acknowledgements}
SCA is supported by an NSERC Discovery Grant.
BW thanks Brock University for support during the period when this work was completed.

On behalf of all authors, the corresponding author states that there is no conflict of interest. 

The reviewer is thanked for remarks which have improved parts of the paper.

\end{document}